\documentclass[%
 reprint,
 amsmath,amssymb,
 aps,
 nofootinbib,
 superscriptaddress
]{article}

\usepackage{amsthm}
\usepackage[dvipdfmx]{graphicx}
\usepackage[dvipdfmx]{xcolor}
\usepackage{dcolumn}
\usepackage{bm}
\usepackage{braket}
\usepackage{amsmath}
\usepackage{amssymb}
\usepackage{enumerate}
\usepackage{booktabs}
\usepackage{txfonts}
\usepackage[update,prepend]{epstopdf}
\usepackage[english]{babel}
\usepackage{listings}
\usepackage[subrefformat=parens]{subcaption}
\usepackage{caption}
\usepackage{multirow}
\usepackage{authblk}
\usepackage{cite}
\usepackage[a4paper]{geometry}
\usepackage{algorithmic}
\usepackage{algorithm}

\newtheorem{theorem}{Theorem}[section]
\newtheorem{lemma}{Lemma}[section]
\newtheorem{assum}{Assumption}[section]
\newtheorem{problem}{Problem}

\begin{document}

\title{Pricing multi-asset derivatives by finite difference method on a quantum computer}

\renewcommand{\thefootnote}{\fnsymbol{footnote}}

\author{Koichi Miyamoto$^{1}$\footnote{koichi.miyamoto@qiqb.osaka-u.ac.jp} \ and Kenji Kubo$^{2,3}$\footnote{kenjikun@mercari.com}\\

	\footnotesize $^{1}$ Center for Quantum Information and Quantum Biology, Osaka University,\\ 1-3 Machikaneyama, Toyonaka, Osaka, 560-8531, Japan \\	
	\footnotesize $^{2}$ R4D, Mercari Inc., Roppongi Hills Mori Tower 18F,\\ 6-10-1, Roppongi, Minato-ku, Tokyo 106-6118, Japan \\
	\footnotesize $^{3}$ Graduate School of Engineering Science, Osaka University,\\ 1-3, Machikaneyama, Toyonaka, Osaka, 560-8531, Japan \\
}

\date{\today}
\maketitle

\renewcommand{\thefootnote}{\arabic{footnote}}

\begin{abstract}
Following the recent great advance of quantum computing technology, there are growing interests in its applications to industries, including finance.
In this paper, we focus on derivative pricing based on solving the Black-Scholes partial differential equation by finite difference method (FDM), which is a suitable approach for some types of derivatives but suffers from the {\it curse of dimensionality}, that is, exponential growth of complexity in the case of multiple underlying assets.
We propose a quantum algorithm for FDM-based pricing of multi-asset derivative with exponential speedup with respect to dimensionality compared with classical algorithms.
The proposed algorithm utilizes the quantum algorithm for solving differential equations, which is based on quantum linear system algorithms.
Addressing the specific issue in derivative pricing, that is, extracting the derivative price for the present underlying asset prices from the output state of the quantum algorithm, we present the whole of the calculation process and estimate its complexity.
We believe that the proposed method opens the new possibility of accurate and high-speed derivative pricing by quantum computers.
\end{abstract}

\section{Introduction \label{sec:Intro}}
Recently, people are witnessing the great advance of {\it quantum computing}\footnote{For readers who are unfamiliar to quantum computing, we refer to \cite{NC} as a standard textbook.}, which can speedup some computational tasks compared with exiting {\it classical} computers, and taking a strong interest in its industrial applications.
Finance is one of promising fields.
Since large financial institutions perform enormous computational tasks in their daily business\footnote{For readers who are unfamiliar to financial engineering or, more specifically, derivative pricing, we refer to \cite{Hull,Shreve}.}, it is naturally expected that quantum computers will tremendously speedup them and make a large impact on the industry.
In fact, some recent papers have already discussed applications of {\it quantum algorithms} to concrete problems in financial engineering: for example, derivative pricing\cite{Rebentrost,Martin,Stamatopoulos,Ramos-Calderer,Fontanela,Vazquez,Kaneko,Tang,Chakrabarti,An,Gonzalez-Conde,Radha,Alghassi}, risk measurement\cite{Worner,Egger,Miyamoto}, portfolio optimization\cite{Rebentrost2,Kerenidis,Hodson}, and so on.
See \cite{Orus,Egger2,Bouland} as comprehensive reviews.

Among a variety of applications, we here consider the quantum method for derivative pricing.
Especially, we focus on the approach based on solving the partial differential equation (PDE) by finite difference method (FDM).
Let us describe the outline of the problem.
First of all, we explain what {\it financial derivatives}, or simply {\it derivatives} are.
They are products whose values are determined by prices of other simple assets ({\it underlying assets}) such as stocks, bonds, foreign currencies, commodities, and so on.
They can be characterized by payoffs paid and/or received between parties involved in a derivative contract, whose amounts are determined by underlying asset prices.
One of the simplest examples of derivatives is an European call (resp. put) option, the right to buy (resp. sell) some asset at the predetermined price ({\it strike}) $K$ and time ({\it maturity}) $T$.
This is equivalent to the contract that the option buyer receives the payoff $f_{\rm pay}(S_T)=\max\{S_T-K,0\}$ (call) or $\max\{K-S_T,0\}$ (put) at $T$, where $S_t$ is the underlying asset price at time $t$.
Besides this, there are many types of derivatives, some of which contain complicated contract terms and are called {\it exotic} derivatives.

Since large banks hold a large number of exotic derivatives, pricing them is crucial for their business.
We can evaluate a derivative price by modeling the random time evolution of underlying asset prices by some stochastic processes and calculating the expected values of the payoff\footnote{Strictly speaking, the payoff must be divided by some numeraire.} under some probability measure.
Since analytical formulas for the derivative price are available only in limited settings, we often resort to numerical methods.
One of major approaches is Monte Carlo simulation.
That is, we generate many paths of underlying asset price evolution on some discretized time grid, and then take the average of payoffs over the paths.
We can also take another approach: since the expected value obeys some PDE, which is called the Black-Scholes (BS) PDE, we can obtain the derivative price by solving it\footnote{For comprehensive reviews of the PDE approach for derivative pricing, we refer to \cite{Tavella,Duffy} as textbooks.}.
More concretely, starting from the maturity $T$, at which the derivative price is trivially determined as the payoff itself, we solve the PDE backward to the present, and then find the present value of the derivative.

We should choose the appropriate approach according to the nature of the problem.
For example, PDE approach is suitable for derivatives whose price is subject to some continuous {\it boundary conditions}.
One prominent example is the {\it barrier option}.
In a barrier option contract, one or multiple levels of underlying asset prices, which are called barriers, are set.
Then, they determine whether the payoff is paid at the maturity or not. 
For example, in a {\it knock-out} barrier option, the payoff is not paid if either of barriers is reached once or more by $T$, regardless of $S_T$\footnote{There are also {\it knock-in} barrier options, where the payoff is paid only if either of barriers is reached at least once by $T$. We can price an knock-in barrier option by subtracting the price of the corresponding knock-out barrier option from that of the corresponding European (that is, no-barrier) option, since a combination of a knock-in barrier option and a knock-out barrier option is equivalent to an European option.}.
This means that the price of the knock-out barrier option is 0 at barriers.
Such a boundary condition is difficult to be strictly taken into account in the Monte Carlo approach because of discretized time evolution, but it can be dealt with in the PDE approach.

Although the PDE approach is suitable in these cases, it is difficult to apply it to {\it multi-asset derivatives}, that is, the case where the number of underlying assets $d$ is larger than 1.
This is because of the exponential growth of complexity with respect to $d$, which is known as the {\it curse of dimensionality}.
We can see this as follows.
The BS PDE is $(d+1)$-dimensional, where $d$ and $1$ correspond to asset prices and time, respectively.
In FDM, which is often adopted for solving a PDE numerically, the discretization grid points are set in the asset price directions, and partial derivatives are replaced with matrices which correspond to finite difference approximation.
This converts a PDE into a linear ordinary differential equation (ODE) system, in which the dependent variables are the derivative prices on grid points and the independent variable is time.
Then, we solve the resulting ODE system. 
The point is that this calculation contains manipulations of the matrices with exponentially large size, that is, $n_{\rm gr}^d \times n_{\rm gr}^d$, where $n_{\rm gr}$ is the number of the grid points in one direction.
Since we have to take $n_{\rm gr}$ proportional to $\epsilon^{-1/2}$ in order to accomplish the error level $\epsilon$, as shown later, the time complexity of this approach grows as $O({\rm poly}(\epsilon^{-d/2}))=O((1/\epsilon)^{{\rm poly}(d)})$.
Besides, the space complexity also grows exponentially, since we have to store the derivative prices on grid points in calculation.
This makes the PDE approach, at least in combination with FDM, intractable on classical computers.

Fortunately, quantum computers might change the situation.
This is because there are some quantum algorithms for solving linear ODE systems, whose time complexities depend on dimensionality only logarithmically\cite{Berry,Berry2,Xin,Childs}.
This means that, in combination with these algorithms, we can remove the exponential dependency of time complexity of FDM on dimensionality.
In fact, some quantum algorithms for solving PDE, including not only FDM-based ones but also different approaches, have already been proposed, and quantum speedup is obtained in some cases\cite{Cao,Clader,Montanaro,Fillion-Gourdeau,Costa,Wang,Childs2,Linden}.
Note also that the space complexity can be also reduced exponentially, since, using a $n$-qubits system, we can encode a vector with exponentially large size with respect to $n$ into the amplitudes of the quantum state.

In light of the above, this paper aims to speedup FDM-based pricing of multi-asset derivatives, utilizing the quantum algorithm.
Although one might think that this is just a straightforward application of an existing algorithm to some problem, there is a nontrivial issue specific for derivative pricing.
The issue is how to extract the present value of the derivative from the output of the quantum algorithm.
By solving the BS PDE up to the present ($t=0$) using the quantum algorithm, we obtain the vector $\vec{V}(0)$, which consists of the derivative prices on the grid points in the space of the underlying asset prices.
However, it is given not as classical data but as a quantum state $\ket{\vec{V}(0)}$, in which the elements of $\vec{V}(0)$ are encoded as amplitudes of computational basis states.
On the other hand, typically, we are interested in only one element of $\vec{V}(0)$, that is, $V_0$, the derivative price for the present underlying asset prices.
This means that we have to obtain the amplitude of the specific computational basis state in $\ket{\vec{V}(0)}$.
Since the amplitude is exponentially small if the number of the grid point is exponentially large, reading it out requires exponentially large time complexity, which ruins the quantum speedup.

We circumvent this issue by solving PDE up to not the present but some future time $t_{\rm ter}$.
The key observation is that $V_0$ can be expressed as the expected value of its price\footnote{Again, strictly speaking, the price must be divided by some numeraire.} at an arbitrary future time.
Concretely, we may take the following way.
First, we generate two states: $\ket{\vec{V}(t_{\rm ter})}$, in which the derivative prices at $t_{\rm ter}$ are encoded, and $\ket{\vec{p}(t_{\rm ter})}$, in which the probability distribution of underlying asset prices at $t_{\rm ter}$ are encoded.
Then, we estimate the inner product $\braket{\vec{p}(t_{\rm ter})|\vec{V}(t_{\rm ter})}$, which is an approximation of $V_0$.
Note that the amplitude of each basis state in $\ket{\vec{V}(t_{\rm ter})}$ contains information to determine $V_0$ differently from $\ket{\vec{V}(0)}$, in which the amplitude of one specific basis state is the sole necessary information.
This leads to much smaller time complexity in the above way than reading $V_0$ out from $\ket{\vec{V}(0)}$.

In the following sections, we describe the entire process of the above calculation: setting $t_{\rm ter}$, generating $\ket{\vec{V}(t_{\rm ter})}$ by the quantum algorithm, generating $\ket{\vec{p}(t_{\rm ter})}$, and estimating $V_0$.
Besides, we estimate the complexity of the proposed method.
We see that, in the expression of the complexity, there are not any factors like $(1/\epsilon)^{{\rm poly}(d)}$ but only some logarithmic factors to the power of $d$, which means substantial speedup compared with classical FDM.

The rest of this paper is organized as follows.
Sections \ref{sec:FDM} and \ref{sec:QAlgODE} are preliminary ones, which outline derivative pricing based on solving a PDE by FDM and the quantum algorithm for solving ODE systems, respectively.
In Section \ref{sec:app}, we discuss approximating $V_0$ as the expected value of the price at $t_{\rm ter}$.
Here, we also discuss how to set $t_{\rm ter}$, taking into account the probability that underlying assets reach the barrier.
Section \ref{sec:QFDM} presents the main result, that is, the quantum calculation procedure for $V_0$ and its complexity.
Section \ref{sec:Sum} summarizes this paper.
All proofs are shown in the appendix.

\subsection{Notations}

Here, we explain the notations used in this paper.

$\mathbb{R}_+$ means the set of all positive real numbers: $\mathbb{R}_+:=\{x\in \mathbb{R}|x>0\}$.
For a positive integer $d$, $\mathbb{R}_+^d$ is its $d$-times direct products: $\mathbb{R}_+^d:=\underbrace{\mathbb{R}_+ \times \cdots \times \mathbb{R}_+}_d$.
For a positive integer $n$, $[n]:=\{1,...,n\}$.

For a positive integer $n$, $I_n$ denotes the $n\times n$ identity matrix.
$\|\cdot\|$ means the Euclidian norm for a vector and the spectral norm for a matrix.
We call each of them a ``norm" simply.
For a $n\times n$ matrix $A$, $\mu(A)$ means the logarithmic norm associated with $\|\cdot\|$: $\mu(A):=\lim_{h\rightarrow 0^+}(\|I_n+hA\|-1)/h$.
When a matrix $A$ has at most $s$ nonzero entries in any row and column, we say that the sparsity of $A$ is $s$.

In this paper, we consider quantum states of systems consisting of some quantum registers with some qubits.
For a real number $x$, $\ket{x}$ denotes one of the computational basis states on some register, whose bit string corresponds to the binary representation of $x$.
For $i\in\{0,1\}$, we let $\ket{i}$ and $\ket{\bar{i}}$ denote a state on a multi-qubit register and a state on one qubit, respectively, in order to distinguish them.
For $\vec{x}:=(x_1,...,x_d)^T\in \mathbb{R}^d$, $\ket{\vec{x}}$ denotes the (unnormalized) state in which the elements of $\vec{x}$ are encoded in the amplitudes of computational basis states, that is, $\ket{\vec{x}}:=\sum_{i=1}^dx_i\ket{i}$.
For a (unnormalized) state $\ket{\psi}$, its norm is defined as $\|\ket{\psi}\|:=\sqrt{\braket{\psi|\psi}}$.
If a state $\ket{\psi}$ satisfies $\|\ket{\psi}-\ket{\psi^\prime}\|<\epsilon$, where $\epsilon$ is a positive real number and $\ket{\psi^\prime}$ is another state, we say that $\ket{\psi}$ is $\epsilon$-close to $\ket{\psi^\prime}$.

\section{Derivative pricing based on solving the PDE by FDM \label{sec:FDM}}

\subsection{Derivative pricing problem and the Black-Scholes PDE}

In this paper, we consider the following problem.

\begin{problem}
	Let $d$ be a positive integer and $T,U_1,...,U_d,L_1,...,L_d$ be positive real numbers such that $L_i<U_i$ for $i\in[d]$.
	Define $D:=(L_1,U_1)\times\cdots\times(L_d,U_d),\bar{D}:=[L_1,U_1]\times\cdots\times[L_d,U_d]$ and $\hat{D}^{i}:=[L_1,U_1]\times\cdots\times[L_{i-1},U_{i-1}]\times[L_{i+1},U_{i+1}]\times\cdots\times[L_d,U_d]$ for $i\in[d]$.
	Assume that a function $V:[0,T]\times \bar{D}\rightarrow \mathbb{R}$ satisfies the following PDE
	\begin{equation}
		\frac{\partial}{\partial t}V(t,\vec{S}) + \frac{1}{2}\sum_{i,j=1}^d \sigma_i\sigma_j\rho_{ij}S_iS_j \frac{\partial^2}{\partial S_i\partial S_j}V(t,\vec{S}) + r\left(\sum_{i=1}^d S_i\frac{\partial}{\partial S_i}V(t,\vec{S})-V(t,\vec{S})\right)=0, \label{eq:PDE}
	\end{equation}
	on $[0,T)\times D$ and boundary conditions
	\begin{eqnarray}
		&& V(T,\vec{S}) = f_{\rm pay}(\vec{S}), \nonumber \\
		&&V(t,(S_1,...,S_{i-1},U_i,S_{i+1},...,S_d)^T) = V^{\rm UB}_i(t,(S_1,...,S_{i-1},S_{i+1},...,S_d)^T) \ {\rm for} \ i\in[d], \nonumber \\
		&&V(t,(S_1,...,S_{i-1},L_i,S_{i+1},...,S_d)^T) = V^{\rm LB}_i(t,(S_1,...,S_{i-1},S_{i+1},...,S_d)^T) \ {\rm for} \ i\in[d]. \label{eq:BounCond}
	\end{eqnarray}
	Here, $t\in[0,T]$, $\vec{S}:=(S_1,...,S_d)^T\in D$, $\sigma_1,...,\sigma_d,r$ are positive real constants such that $r<\frac{1}{2}\sigma_i^2$ for $i\in[d]$, $\rho_{ij},i,j\in[d]$ are real constants such that $\rho_{11}=\cdots=\rho_{dd}=1$ and the matrix $\rho:=(\rho_{ij})_{\substack{1\le i\le d \\ 1\le j \le d}}$ is symmetric and positive-definite, and $f_{\rm pay}:D\rightarrow \mathbb{R}$, $V^{\rm UB}_i:[0,T]\times \hat{D}^{i}\rightarrow \mathbb{R}$ and $V^{\rm LB}_i:[0,T]\times \hat{D}^{i}\rightarrow \mathbb{R}$ are given functions.
	Then, for a given $\vec{S}_0:=(S_{1,0},...,S_{d,0})^T\in D$, find $V_0:=V(0,\vec{S}_0)$.
\end{problem}

Here, we make some comments.
(\ref{eq:PDE}) is the so-called BS PDE, which corresponds to the following derivative pricing problem.
Under some probability space $(\Omega,\mathcal{F},P)$, we consider the $d$-dimensional stochastic process $\vec{S}(t):=(S_1(t),...,S_d(t))^T,t\ge 0$ obeying the following stochastic differential equation (SDE) system:
\begin{equation}
	dS_i(t) = rS_i(t)dt + \sigma_i S_i(t) dW_i(t), i\in[d] \label{eq:SDE}
\end{equation}
where $W_1,...,W_d$ are the Brownian motions on $(\Omega,\mathcal{F},P)$ satisfying $dW_idW_j=\rho_{ij}dt$ for $i,j\in[d]$ and the initial value is $\vec{S}(0)=\vec{S}_0$.
$S_1,...,S_d$ correspond to prices of $d$ underlying assets and (\ref{eq:SDE}) describes the random time evolution of $\vec{S}(t)$ under the so-called risk-neutral measure, where any asset price grows with the risk-free rate $r$ in expectation.
$\sigma_i$ is the parameter called {\it volatility}, which parameterizes how the random movement of $S_i$ is volatile.
This is the so-called BS model.
Then, the derivative price is given by the conditional expected value of the payoff discounted by the risk-free rate.
That is, the price of the derivative in which the payoff $f_{\rm pay}(\vec{S}(T))$ arises at maturity $T$ is
\begin{equation}
	V(t,\vec{S})=E[e^{-r(T-t)}f_{\rm pay}(\vec{S}(T))1_{\rm NB}|\vec{S}(t)=\vec{S}]
\end{equation}
at time $t$, if $\vec{S}(t)=\vec{S}$.
Here,
$1_{\rm NB}$ is a stochastic variable taking 1 if the condition for the payoff to be paid (e.g., barrier condition) is satisfied or 0 otherwise.
It is known that $V(t,\vec{S})$ satisfies (\ref{eq:PDE}) and appropriate boundary conditions, which should be set according to the product characteristics of the derivative such as barrier conditions\cite{Duffy,Tavella}.
We here present some typical choices:
\begin{itemize}
	\item If $U_i$ is a knock-out barrier, that is, $S_i$ reaching $U_i$ leads to the payoff not being paid, $V^{\rm UB}_i=0$. Similarly, if $L_i$ is a knock-out barrier, $V^{\rm LB}_i=0$.
	\item Suppose that $f_{\rm pay}(\vec{S}(T))$ takes the form of $\max\{a_0+\sum_{i=1}^d a_iS_i(T),0\}$ with $a_0,a_1,...,a_d\in\mathbb{R}$, which is the case with many types of derivatives including call and put options.
	In such a case, when either of $S_i$'s is extremely high or low, the derivative can be far in-the-money, which means that it is highly likely that the positive payoff will be paid (e.g., $a_0<0$, $a_1,...,a_d>0$ (i.e., a basket call option) and $S_i\gg -a_0/a_i$).
	In this situation, the derivative price is nearly equal to the discounted payoff.
	Therefore, we can set
	\begin{equation}
		V^{\rm UB}_i(t,(S_1,...,S_{i-1},S_{i+1},...,S_d)^T)=e^{-r(T-t)}\left(a_0+\sum_{\substack{1\le j\le d \\ j\ne i}} a_jS_j + a_iU_i\right)
	\end{equation}
	for sufficiently large $U_i$.
	In some cases, $V^{\rm LB}_i$ can be set in the similar way.
\end{itemize}

For a later convenience, we here transform the PDE (\ref{eq:PDE}) on $[0,T)\times D$ into
\begin{eqnarray}
	&&\frac{\partial}{\partial \tau}Y(\tau,\vec{x}) = \mathcal{L}Y(\tau,\vec{x}) \nonumber \\
	&&\mathcal{L} := \frac{1}{2}\sum_{i,j=1}^d \sigma_i\sigma_j\rho_{ij} \frac{\partial^2}{\partial x_i\partial x_j} + \sum_{i=1}^d\left(r-\frac{1}{2}\sigma_i^2\right)\frac{\partial}{\partial x_i}, \label{eq:PDE2}
\end{eqnarray}
on $(0,T]\times \Tilde{D}$, where $\tau:=T-t,\vec{x}:=(x_1,...,x_d)^T:=(\log S_1,...,\log S_d)^T$, $Y(\tau,\vec{x}):=e^{r\tau}V(T-\tau,(e^{x_1},...,e^{x_d})^T)$, $\Tilde{D}:=(l_1,u_1)\times\cdots\times(l_d,u_d)$ and $u_i:=\log U_i,l_i:=\log L_i$ for $i\in[d]$.
The boundary conditions become
\begin{align}
	&Y(0,\vec{x}) = \Tilde{f}_{\rm pay}(\vec{x}) := f_{\rm pay}((e^{x_1},...,e^{x_d})^T), \nonumber & \\
	&Y(\tau,(x_1,...,x_{i-1},u_i,x_{i+1},...,x_d)^T)= Y^{\rm UB}_i(\tau,(x_1,...,x_{i-1},x_{i+1},...,x_d)^T) :=V^{\rm UB}_i(T-\tau,(e^{x_1},...,e^{x_{i-1}},e^{x_{i+1}},...,e^{x_{d}})^T) \nonumber & \\
	&\qquad\qquad\qquad\qquad\qquad\qquad\qquad\qquad\qquad\qquad\qquad\qquad\qquad\qquad\qquad\qquad\qquad\qquad\qquad\qquad {\rm for} \ i\in[d],& \nonumber\\
	&Y(\tau,(x_1,...,x_{i-1},l_i,x_{i+1},...,x_d)^T) =Y^{\rm LB}_i(\tau,(x_1,...,x_{i-1},x_{i+1},...,x_d)^T) := V^{\rm LB}_i(T-\tau,(e^{x_1},...,e^{x_{i-1}},e^{x_{i+1}},...,e^{x_{d}})^T) \nonumber & \\
	& \qquad\qquad\qquad\qquad\qquad\qquad\qquad\qquad\qquad\qquad\qquad\qquad\qquad\qquad\qquad\qquad\qquad\qquad\qquad\qquad{\rm for} \ i\in[d].&
	\label{eq:BounCond2}
\end{align}

\subsection{Application of FDM to the BS PDE}

FDM is a method for solving a PDE by replacing partial derivatives with finite difference approximations.
In the case of (\ref{eq:PDE2}), the approximation is as follows.
First, letting $n_{\rm gr}$ be a positive integer, we introduce the grid points in the directions of $\vec{x}$:
\begin{eqnarray}
	\vec{x}^{(k)} &:=& (x^{(k_1)}_{1},...,x^{(k_d)}_{d})^T, k=\sum_{i=1}^d n_{\rm gr}^{d-i}k_i+1,k_i=0,1,...,n_{\rm gr}-1\nonumber \\
	x^{(k_i)}_{i} &:=& l_i + (k_i+1)h_i\nonumber \\
	h_i &:=& \frac{u_i-l_i}{n_{\rm gr}+1}.
\end{eqnarray}
Namely, there are $n_{\rm gr}$ equally spaced grid points in one direction and the total number of the grid points in $D$ is $N_{\rm gr}:=n_{\rm gr}^d$, except ones on the boundaries.
For later convenience, we set $x^{(-1)}_{i}=l_1$ and $x^{(n_{\rm gr})}_{i}=h_1$.
Hereafter, we assume that $n_{\rm gr}$ is a power of 2 for simplicity, whose detail is explained in Section \ref{sec:QFDM}, and define $m_{\rm gr}:=\log_2 n_{\rm gr}$.

Then, (\ref{eq:PDE2}) is transformed into the $N_{\rm gr}$-dimensional ODE system
\begin{equation}
	\frac{d}{d\tau}\vec{\Tilde{Y}}(\tau)=F \vec{\Tilde{Y}}(\tau)+\vec{C}(\tau). \label{eq:ODE}
\end{equation}
with the initial value
\begin{equation}
	\vec{\Tilde{Y}}(0) = (Y(0,\vec{x}^{(1)}),...,Y(0,\vec{x}^{(N_{\rm gr})}))^T = (\Tilde{f}_{\rm pay}(\vec{x}^{(1)}),...,\Tilde{f}_{\rm pay}(\vec{x}^{(N_{\rm gr})}))^T =: \vec{\Tilde{f}}_{\rm pay}. \label{eq:IniODE}
\end{equation}
Here, $\vec{\Tilde{Y}}(\tau), F$ and $\vec{C}(\tau)$, which newly appear in (\ref{eq:ODE}), are as follows.
$\vec{\Tilde{Y}}(\tau):=(\Tilde{Y}_1(\tau),...,\Tilde{Y}_{N_{\rm gr}}(\tau))\in \mathbb{R}^{N_{\rm gr}}$ and its $k$-th element is an approximation of $Y(\tau,x^{(k)})$.
$F$ is a $N_{\rm gr}\times N_{\rm gr}$ real matrix, which is expressed by a sum of Kronecker products of $n_{\rm gr}\times n_{\rm gr}$ matrices, that is,
\begin{eqnarray}
	F &:=& F^{\rm 2nd} + F^{\rm 1st} \nonumber \\
	F^{\rm 2nd} &:=& \sum_{i=1}^d \frac{\sigma_i^2}{2h_i^2} I^{\otimes i-1} \otimes D^{\rm 2nd} \otimes I^{\otimes d-i} + \sum_{i=1}^{d-1} \sum_{j=i+1}^d  \frac{\sigma_i\sigma_j\rho_{ij}}{4h_ih_j} I^{\otimes i-1} \otimes D^{\rm 1st} \otimes I^{\otimes j-i-1} \otimes D^{\rm 1st} \otimes I^{\otimes d-j} \nonumber \\
	F^{\rm 1st}&:=& \sum_{i=1}^d \frac{1}{2h_i}\left(r-\frac{1}{2}\sigma_i^2\right) I^{\otimes i-1} \otimes D^{\rm 1st} \otimes I^{\otimes d-i}, \label{eq:F}
\end{eqnarray}
where $I$ is the $n_{\rm gr}\times n_{\rm gr}$ identity matrix and 
\begin{equation}
	D^{\rm 1st} :=
	\begin{pmatrix}
		0 & 1  &    &   & & \\
		-1 & 0  & 1  &   & & \\
		& -1 & 0  & 1 & & \\
		&    & \ddots  & \ddots & \ddots & \\
		&    & & -1 & 0 & 1\\
		&    & &  & -1 & 0
	\end{pmatrix},
	D^{\rm 2nd} :=
	\begin{pmatrix}
		-2 & 1  &    &   & & \\
		1 & -2  & 1  &   & & \\
		& 1 & -2 & 1 & & \\
		&    & \ddots  & \ddots & \ddots & \\
		&    & & 1 & -2 & 1\\
		&    & &  & 1 & -2
	\end{pmatrix}
\end{equation}
are $n_{\rm gr}\times n_{\rm gr}$ tridiagonal matrices.
$\vec{C}(\tau):=(C_1(\tau),...,C_{N_{\rm gr}}(\tau))^T$ is necessary to take into account the boundary conditions and its $k$-th element is
\begin{eqnarray}
	C_k(\tau) &=& \sum_{i=1}^d\frac{\sigma_i^2}{2h_i^2}\left[\delta_{k_i,0}Y^{\rm LB}_i(\tau,\vec{x}^{(k)})+\delta_{k_i,n_{\rm gr}-1}Y^{\rm UB}_i(\tau,\vec{x}^{(k)})\right] \nonumber \\
	&+& \sum_{i=1}^{d-1} \sum_{j=i+1}^d  \frac{\sigma_i\sigma_j\rho_{ij}}{4h_ih_j}\left[-\delta_{k_i,0}Y^{\rm LB}_i(\tau,\vec{x}^{(k)})-\delta_{k_j,0}Y^{\rm LB}_j(\tau,\vec{x}^{(k)})
	+\delta_{k_i,n_{\rm gr}-1}Y^{\rm UB}_i(\tau,\vec{x}^{(k)})+\delta_{k_j,n_{\rm gr}-1}Y^{\rm UB}_j(\tau,\vec{x}^{(k)})\right] \nonumber \\
	&+&\sum_{i=1}^d\frac{1}{2h_i}\left(r-\frac{1}{2}\sigma_i^2\right)\left[\delta_{k_i,n_{\rm gr}-1}Y^{\rm UB}_i(\tau,\vec{x}^{(k)})-\delta_{k_i,0}Y^{\rm LB}_i(\tau,\vec{x}^{(k)})\right]. \label{eq:C}
\end{eqnarray}

Then, let us discuss the accuracy of the approximation (\ref{eq:ODE}).
First, we make a following assumption.
\begin{assum}
	$Y(\tau,\vec{x})$, the solution of (\ref{eq:PDE2}) and (\ref{eq:BounCond2}), is four-times differentiable with respect to $x_1,...,x_d$ and there exist $\zeta,\xi\in\mathbb{R}$ such that
	\begin{equation}
		\forall i,j,k,l\in [d], \tau\in(0,T), \vec{x}\in \Tilde{D},
		 \left|\frac{\partial^3 Y}{\partial x_i \partial x_j \partial x_k}(\tau,\vec{x})\right| < \zeta, \left|\frac{\partial^4 Y}{\partial x_i \partial x_j \partial x_k \partial x_l}(\tau,\vec{x})\right| < \xi. \label{eq:UDerivCond}
	\end{equation}
	\label{ass:YSm}
\end{assum}
\noindent We then obtain the following lemma, as proved in Appendix \ref{sec:PrLemFDMErr}.
\begin{lemma}
	Let $Y(\tau,\vec{x})$ be the solution of (\ref{eq:PDE2}) and (\ref{eq:BounCond2}), and $\vec{\Tilde{Y}}(\tau)$ be that of (\ref{eq:ODE}) and (\ref{eq:IniODE}).
	Under Assumption \ref{ass:YSm}, if, for a given $\epsilon\in\mathbb{R}_+$,
	\begin{equation}
		h_i < \min \left\{\frac{1}{d\sigma_i}\sqrt{\frac{3\epsilon}{2\xi T}},\frac{1}{\sigma_i}\sqrt{\frac{3\epsilon}{\zeta dT}}\right\}, i\in[d] \label{eq:hxcond2}
	\end{equation}
	then, for any $\tau\in (0,T)$, the inequality
	\begin{equation}
		\|\vec{\Tilde{Y}}(\tau)-\vec{Y}(\tau)\| < \sqrt{N_{\rm gr}}\epsilon \label{eq:tilUDiff}
	\end{equation}
	holds, where $\vec{Y}(\tau)=(Y(\tau,\vec{x}^{(1)}),...,Y(\tau,\vec{x}^{(N_{\rm gr})}))^T$.
\label{lem:FDMErr}
\end{lemma}
\noindent 
Lemma \ref{lem:FDMErr} means that the root mean square of the differences between $\Tilde{Y}_i(\tau)$ and $Y(\tau,\vec{x}^{(i)})$ is upper bounded by $\epsilon$.
This result will be reflected to the estimation of the error in the proposed method for Problem 1.

\section{Quantum algorithm for solving ordinary differential equation systems\label{sec:QAlgODE}}

In this section, we outline the algorithm of \cite{Berry2}.
This is the algorithm for solving the linear ODE system
\begin{equation}
	\frac{d}{dt}\vec{x}(t)=A\vec{x}+\vec{b}, \label{eq:ODEEx}
\end{equation}
with the initial condition $\vec{x}(0)=\vec{x}_{\rm ini}$.
Here, $\vec{x}(t)\in\mathbb{R}^N$, $A\in \mathbb{R}^{N\times N}$ is a constant diagonalizable matrix, and $\vec{b}\in \mathbb{R}^N$ is a constant vector.
Suppose that we want to find $\vec{x}(T)$ for some $T\in\mathbb{R}_+$.
The algorithm is based on the formal solution of (\ref{eq:ODEEx})
\begin{equation}
	\vec{x}(T)=e^{AT}\vec{x}_{\rm ini}+(e^{AT}-I_N)A^{-1}\vec{b}. \label{eq:formSol}
\end{equation}
In order to calculate this, we consider the linear equation system on the tensor product space $V:=\mathbb{R}^{q+1} \otimes \mathbb{R}^N$, where the former is the auxiliary space and the latter is the original space on which $A$ operates:
\begin{equation}
	C_{m,k,p}(Ah_t)\vec{X}=\vec{e}_0\otimes\vec{x}_{\rm ini}+h\sum_{i=0}^{m-1}\vec{e}_{i(k+1)+1}\otimes\vec{b}. \label{eq:LSBerry}
\end{equation}
Here, $m,p,k$ are positive integers set large enough (see the statement of Theorem \ref{th:QODE}), $q:=m(k+1)+p$, $h_t=T/m$, $\vec{X}\in\mathbb{R}^{N(q+1)}$ and $\{\vec{e}_i\}_{i=0,1,...,q}$ is an orthonormal basis of $\mathbb{R}^{q+1}$.
For $B\in \mathbb{R}^{N\times N}$, the $N(q+1)\times N(q+1)$ matrix $C_{m,k,p}(B)$ is defined as
\begin{eqnarray}
	C_{m,k,p}(B)&:=&\sum_{j=0}^{q}\vec{e}_j\vec{e}_j^T\otimes I_N-\sum_{i=0}^{m-1}\sum_{j=1}^k\vec{e}_{i(k+1)+j}\vec{e}_{i(k+1)+j-1}^T\otimes \frac{1}{j}B \nonumber \\
	&& \quad -\sum_{i=0}^{m-1}\sum_{j=0}^k\vec{e}_{(i+1)(k+1)}\vec{e}_{i(k+1)+j}^T\otimes I_N - \sum_{j=m(k+1)+1}^{q}\vec{e}_j\vec{e}_{j-1}^T\otimes I.
\end{eqnarray}
Visually, (\ref{eq:LSBerry}) is displayed as follows
\begin{equation}
	\begin{pmatrix}
		I_N       &        &         &           & & & &  & & & & \\
		-Ah_t/1 & I_N      &         &           &&&&&&&& \\
		& \ddots & \ddots  &           &&&&&&&& \\
		&        & -Ah_t/k & I_N         &&&&&&&& \\
		-I_N      & \cdots & -I_N      & -I_N     & I_N      &         &&&&&&\\
		&        &         &        & \ddots & \ddots  & &&&&&\\
		&        &         &        &        & -Ah_t/1 & I_N &&&&&\\
		&		 &         &        &        &         & \ddots & \ddots  & &&& \\
		&		 &         &        &        &         &        & -Ah_t/k & I_N      &&& \\
		&		 &         &        &        & -I_N      & \cdots & -I_N      & -I_N     & I_N       & & \\
		&		 &         &        &        &         &        &         &        & -I_N      & I_N      & \\
		&&		 &         &        &        &         &        &         &        &          \ddots & \ddots & \\
		&		 &         &       & &        &         &        &         &        &          & -I_N & I_N 
	\end{pmatrix}
	\vec{X}=
	\begin{pmatrix}
		\vec{x}_{\rm ini} \\
		h_t\vec{b}\\
		0\\
		\vdots\\
		0\\
		\vdots\\
		h_t\vec{b}\\
		0\\
		\vdots\\
		0\\
		0\\
		\vdots\\
		0\\
	\end{pmatrix}.
\end{equation}
$C_{m,k,p}$ is designed based on the Taylor expansion of (\ref{eq:formSol}).
The solution of (\ref{eq:LSBerry}) can be written as
\begin{equation}
	\vec{X}=\sum_{i=0}^{m-1}\sum_{j=1}^k\vec{e}_{i(k+1)+j}\otimes\vec{x}_{i,j} + \sum_{j=0}^p\vec{e}_{m(k+1)+j}\otimes\vec{x}_{m},
\end{equation}
for some vectors $\vec{x}_{i,j},\vec{x}_{m}\in \mathbb{R}^N$, and $\vec{x}_{m}$ becomes close to $\vec{x}(T)$, which we want to find.
Note that $\vec{x}_{m}$ is repeated $p$ times in the solution $\vec{X}$, which enhances the probability of obtaining the desired vector in the output quantum state of the algorithm.

Although the $C_{m,k,p}(Ah_t)$ is an extremely large matrix, the quantum algorithms for solving linear equation systems (QLS algorithms)\cite{HHL,Ambainis,Clader,Childs3} can output the solution of (\ref{eq:LSBerry}) only with complexity of $O(\log \mathcal{N})$, where $\mathcal{N}$ is the number of rows (or columns) in $C_{m,k,p}(Ah_t)$.
The quantum algorithm in \cite{Berry2} leverages the algorithm in \cite{Childs3}.
In order to use it, \cite{Berry2} assumes that the following oracles (i.e. unitary operators) are available:
\begin{itemize}
	\item $O_{A,1}$\\
	For the matrix $A$, given a row index $j$ and an integer $l$, this return $\nu(j,l)$, the column index of the $l$-th nonzero entry in the $j$-th row:
	\begin{equation}
		O_{A,1}:\ket{j}\ket{l} \mapsto \ket{j}\ket{\nu(j,l)}
	\end{equation}
	\item $O_{A,2}$\\
	For the matrix $A$, given a row index $j$ and a column index $k$, this return the $(j,k)$ entry:
	\begin{equation}
		O_{A,2}:\ket{j}\ket{k}\ket{z} \mapsto \ket{j}\ket{k}\ket{z\oplus A_{jk}}
	\end{equation}
	\item $O_{\vec{x}_{\rm ini}}$\\
	This prepares $\frac{1}{\|\vec{x}_{\rm ini}\|}\ket{\vec{x}_{\rm ini}}$ under the control by another qubit:
	\begin{equation}
		O_{\vec{x}_{\rm ini}}:
		\begin{cases}
		\ket{\bar{0}}\ket{0} \mapsto \frac{1}{\|\vec{x}_{\rm ini}\|}\ket{\bar{0}}\ket{\vec{x}_{\rm ini}} \\
		\ket{\bar{1}}\ket{\psi} \mapsto \ket{\bar{1}}\ket{\psi} \ {\rm for \ any} \ \ket{\psi}
		\end{cases}
	\end{equation}
	\item $O_{\vec{b}}$\\
	When $\vec{b}\ne 0$, this prepares $\frac{1}{\|\vec{b}\|}\ket{\vec{b}}$ under the control by another qubit:
	\begin{equation}
	O_{\vec{b}}:
	\begin{cases}
		\ket{\bar{0}}\ket{\psi} \mapsto \ket{\bar{0}}\ket{\psi} \ {\rm for \ any} \ \ket{\psi} \\
		\ket{\bar{1}}\ket{0} \mapsto \frac{1}{\|\vec{b}\|}\ket{\bar{1}}\ket{\vec{b}}
	\end{cases}.
	\end{equation}
	When $\vec{b}= 0$, this is an identity operator.
\end{itemize}

Then, we present the theorem (Theorem 9 in \cite{Berry2}), which states the query complexity of the algorithm, with a slight modification.
\begin{theorem}
	(Theorem 9 in \cite{Berry2}, slightly modified)
	Suppose $A = V DV^{-1}$ is an $N \times N$ diagonalizable matrix, where $D = {\rm diag}(\lambda_0,\lambda_1,...,\lambda_{N-1})$ satisfies ${\rm Re}(\lambda_j) \le 0$ for any $j \in {0,1,...,N-1}$.
	In addition, suppose $A$ has at most $s$ nonzero entries in any row and column,
	and we have oracles $O_{A,1},O_{A,2}$ as above.
	Suppose $\vec{x}_{\rm ini}$ and $\vec{b}$ are $N$-dimensional vectors with known norms
	and we have oracles $O_{\vec{x}_{\rm ini}}$ and $O_{\vec{b}}$ as above.
	Let	$\vec{x}$ evolve according to the differential equation (\ref{eq:ODEEx}) with the initial condition $\vec{x}(0) =\vec{x}_{\rm ini}$.
	Let $T > 0$ and	$g := \max_{t\in[0,T]}	\|\vec{x}(t)\|/\|\vec{x}(T)\|$.
	Then there exists a quantum algorithm that produces a state $\ket{\Tilde{\Psi}}$, which is $\epsilon$-close to
	\begin{equation}
		\ket{\Psi} := \frac{1}{\sqrt{\braket{\Psi_{\rm gar}|\Psi_{\rm gar}}+(p+1)\|\vec{x}(T)\|^2}}\left(\ket{\Psi_{\rm gar}}+\sum_{j=p(k+1)}^{p(k+2)}\ket{j}\ket{\vec{x}(T)}\right) \label{eq:outputBerryAlgo}
	\end{equation}
	using
	\begin{equation}
	O\left(\kappa_V sT\|A\|\times {\rm poly}\left(\log\left(\frac{\kappa_Vs T\|A\|}{\epsilon}\right)\right)\right) \label{eq:compQODE}
	\end{equation}
	queries to $O_{A,1}$, $O_{A,2}$, $O_x$, and $O_b$.
	Here, $\kappa_V = \|V\|\cdot\|V^{-1}\|$	is the condition number of $V$, $p=\lceil T\|A\| \rceil$, $k=\lfloor 2\log\Omega/\log(\log\Omega)\rfloor$, $\Omega=70g\kappa_Vp^{3/2}(\|\vec{x}_{\rm ini}\|+T\|\vec{b}\|)/\epsilon\|\vec{x}(T)\|$, and $\ket{\Psi_{\rm gar}}$ is an unnormalized state which takes the form of $\ket{\Psi_{\rm gar}}=\sum_{j=0}^{p(k+1)-1}\ket{j}\ket{\psi_j}$ with some unnormalized states $\ket{\psi_0},\ket{\psi_1},...,\ket{\psi_{p(k+1)-1}}$ and satisfies $\braket{\Psi_{\rm gar}|\Psi_{\rm gar}}=O(g^2(p+1)\|\vec{x}(T)\|^2)$. \label{th:QODE}
\end{theorem}
\noindent The modifications from Theorem 9 in \cite{Berry2} are as follows.
First, in \cite{Berry2}, it is assumed that we perform post-selection and obtain $\ket{\vec{x}(T)}/\|\ket{\vec{x}(T)}\|$ (strictly speaking, a state close to it).
On the other hand, in Theorem \ref{th:QODE}, the output state is not purely $\ket{\vec{x}(T)}/\|\ket{\vec{x}(T)}\|$ but contains $\ket{\vec{x}(T)}$ as a part in addition to the unnecessary state $\ket{\Psi_{\rm gar}}$.
This is because, in this paper, we use the algorithm of \cite{Berry2} as a subroutine in the quantum amplitude estimation (QAE)\cite{Brassard,Suzuki,Aaronson,Grinko,Nakaji}, as explained in Section \ref{sec:QFDM}, and the iterated subroutine in QAE must be an unitary operation.
This means that we cannot perform post-selection, since it is a non-unitary operation.
Note also that, we do not perform amplitude amplification for $\ket{\Psi_1}$, which is done before post-selection in \cite{Berry2}, and thus a factor $g$, which exists in the expression of the complexity (112) in \cite{Berry2}, has dropped from (\ref{eq:compQODE}) in this paper.
Moreover, the meaning of the closeness $\epsilon$ is different between Theorem \ref{th:QODE} in this paper and Theorem 9 in \cite{Berry2}.
In the former, $\epsilon$ is the closeness between $\ket{\Tilde{\Psi}}$ and $\ket{\Psi}$, which corresponds to $\delta$ in \cite{Berry2}.
On the other hand, Theorem 9 in \cite{Berry2} refers to the closeness of the state after post-selection to $\ket{\vec{x}(T)}/\|\ket{\vec{x}(T)}\|$.
This difference also makes (\ref{eq:compQODE}) different from (112) in \cite{Berry2}. 

\section{Approximating the present derivative price as the expected value of the price at a future time\label{sec:app}}

As we explained in the introduction, we aim to calculate $V_0$ as the expected value of the discounted price at some future time.
Concretely, we set $t_{\rm ter}\in(0,T)$ and calculate
\begin{equation}
	V_0 = e^{-rt_{\rm ter}}\int_{\mathbb{R}_+^d} d\vec{S} \phi(t_{\rm ter},\vec{S})p_{\rm NB}(t_{\rm ter},\vec{S})V(t_{\rm ter},\vec{S}), \label{eq:V0ExpPrice}
\end{equation}
where $\phi(t,\vec{S})$ is the probability density function of $\vec{S}(t)$, and $p_{\rm NB}(t,\vec{S})$ is the conditional probability that the no event which leads to extinction of the payoff happens by $t$ given $\vec{S}(t)=\vec{S}$.
Although (\ref{eq:V0ExpPrice}) holds for any $t_{\rm ter}$, for the effective numerical calculation, $t_{\rm ter}$ should be set carefully.
Recalling our motivation to evade exponential complexity to read out $V_0$, which is explained in Section \ref{sec:Intro}, we want to set $t_{\rm ter}$ as large as possible.
On the other hand, there are some reasons to set $t_{\rm ter}$ small because of existence of boundaries.
First, note that it is difficult to find $p_{\rm NB}(t_{\rm ter},\vec{S})$ explicitly in the multi-asset case.
However, for sufficiently small $t_{\rm ter}$, $p_{\rm NB}(t_{\rm ter},\vec{S})$ is nearly equal to 1, since the payoff is paid at least if $\vec{S}(t)$ does not reach any boundaries and the probability that $\vec{S}(t)$ reaches any boundaries can be neglected for time close to 0.
Besides, note that we obtain the derivative prices only on the points in boundaries by solving PDE.
For small $t_{\rm ter}$, we can approximately calculate $V_0$ using only the information in boundaries, since the probability distribution of $\vec{S}(t_{\rm ter})$ over the boundaries is negligible.
In summary, we should set $t_{\rm ter}$ as large as possible in the range of the value for which the probability distribution of $\vec{S}(t_{\rm ter})$ is almost confined within the boundaries.
For such $t_{\rm ter}$, we can approximate
\begin{equation}
	V_0 \approx e^{-rt_{\rm ter}}\int_{D} d\vec{S} \phi(t_{\rm ter},\vec{S})V(t_{\rm ter},\vec{S}), \label{eq:V0ExpPriceApp}
\end{equation}
or, equivalently,
\begin{equation}
	V_0 \approx e^{-rT}\int_{\Tilde{D}} d\vec{x} \Tilde{\phi}(t_{\rm ter},\vec{x})Y(\tau_{\rm ter},\vec{x}), \label{eq:V0ExpPriceApp2}
\end{equation}
where $\tau_{\rm ter}:=T-t_{\rm ter}$ and $\Tilde{\phi}(t,\vec{x})$ is the probability density of $\vec{x}(t)$ under the BS model (\ref{eq:SDE}) and will be explicitly given later.

Considering the above points, we obtain the lemma, which shows a criterion to set $t_{\rm ter}$.
First, we make an assumption, which is necessary to upper bound the contribution from the outside of the boundaries to the integral (\ref{eq:V0ExpPrice}).

\begin{assum}
	There exist positive constants $A_0,A_1,...,A_d$ such that $f_{\rm pay}$ in Problem 1 satisfies
	\begin{equation}
		f_{\rm pay}(\vec{S})\le \sum_{i=1}^d A_i S_i + A_0 \label{eq:payoffBound}
	\end{equation}
	for any $\vec{S}\in D$. \label{ass:payoffBound}
\end{assum}

\noindent That is, we assume that the payoff is upper bounded by some linear function, which is the case for many cases such as call/put options on linear combinations of $S_1,...,S_d$ (i.e. basket options).
Then, the following lemma holds.

\begin{lemma}
	Consider Problem 1.
	Under Assumption \ref{ass:payoffBound}, for any $\epsilon\in \mathbb{R}_+$ satisfying
	\begin{equation}
		\log\left(\frac{\Tilde{A}d(d+1)}{\epsilon}\right)>\max\left\{\frac{2}{5}\left(1-\frac{2r}{\sigma_i^2}\right)\log\left(\frac{U_i}{S_{i,0}}\right),\frac{2}{5}\left(1-\frac{2r}{\sigma_i^2}\right)\log\left(\frac{S_{i,0}}{L_i}\right)\right\},i\in[d],
		\label{eq:epscond1}
	\end{equation}
	where $\Tilde{A}=\max\{A_1\sqrt{U_1S_{1,0}},...,A_1\sqrt{U_dS_{d,0}},A_0\}$, and
	\begin{equation}
		\epsilon<2d(d+1)\times\max\{A_0,A_1S_{1,0},...,A_dS_{d,0}\}, \label{eq:epscond2}
	\end{equation}
	the inequality
	\begin{equation}
		\left|V(0,\vec{S}_0) - e^{-rT}\int_{\Tilde{\Tilde{D}}} d\vec{x} \Tilde{\phi}(t_{\rm ter},\vec{x})Y(t_{\rm ter},\vec{x})\right| \le 2\epsilon \label{eq:VIntDiff}
	\end{equation}
	holds, where
	\begin{equation}
		t_{\rm ter} := \min \left\{\frac{2\left(\log \left(\frac{U_1}{S_{1,0}}\right)\right)^2}{25\sigma_1^2 \log \left(\frac{2\Tilde{A}d(d+1)}{\epsilon}\right)},...,\frac{2\left(\log \left(\frac{U_d}{S_{d,0}}\right)\right)^2}{25\sigma_d^2 \log \left(\frac{2\Tilde{A}d(d+1)}{\epsilon}\right)},\frac{2\left(\log \left(\frac{S_{1,0}}{L_1}\right)\right)^2}{25\sigma_1^2 \log \left(\frac{2\Tilde{A}d(d+1)}{\epsilon}\right)},...,\frac{2\left(\log \left(\frac{S_{d,0}}{L_d}\right)\right)^2}{25\sigma_d^2 \log \left(\frac{2\Tilde{A}d(d+1)}{\epsilon}\right)}\right\}. \label{eq:tter}
	\end{equation}
	and
	\begin{equation}
		\Tilde{\Tilde{D}} := \left[\frac{1}{2}\left(l_1+x^{(0)}_1\right),\frac{1}{2}\left(x^{(n_{\rm gr}-1)}_1+u_1\right)\right]\times \cdots \times \left[\frac{1}{2}\left(l_d+x^{(0)}_d\right),\frac{1}{2}\left(x^{(n_{\rm gr}-1)}_d+u_d\right)\right].
	\end{equation}
	\label{lem:tter}
\end{lemma}
\noindent The proof is given in Appendix \ref{sec:PrLemtter}.
Note that, in (\ref{eq:VIntDiff}), the region of the integral is slightly different from $\Tilde{D}$, the interior of the boundary in the $\vec{x}$ domain.
This is just for interpreting the finite-sum approximation of the integral as the midpoint rule and explained in the proof of Lemma \ref{lem:finSumApp}.

%
%

\section{Quantum method for derivative pricing by FDM\label{sec:QFDM}}

In this section, we finally present the quantum method for derivative pricing by FDM.
Our idea is calculating the present derivative price $V_0$ as (\ref{eq:V0ExpPrice}), the expected value of the price at the future time $t_{\rm ter}$.
As explained in Section \ref{sec:app}, we approximate (\ref{eq:V0ExpPrice}) as (\ref{eq:V0ExpPriceApp2}).
In fact, we have to approximate (\ref{eq:V0ExpPriceApp2}) further, since we obtain the derivative prices only on the grid points by solving PDE using FDM.
Therefore, we approximate (\ref{eq:V0ExpPriceApp2}) as
\begin{equation}
	V_0\approx e^{-rT}\sum_{k=1}^{N_{\rm gr}}p_k \Tilde{Y}_k(\tau_{\rm ter}),
\end{equation}
where $p_k$ is the existence probability of $\vec{x}(t_{\rm ter})$, the log prices of underlying assets at $t_{\rm ter}$, on the $k$-th grid point and explicitly defined soon.
In other words, we calculate
\begin{equation}
	V_0\approx e^{-rT}\vec{p}\cdot\vec{\Tilde{Y}}(\tau_{\rm ter}), \label{eq:pdotYtil}
\end{equation}
where $\vec{p}:=(p_1,...,p_{N_{\rm gr}})^T$. 
Hereafter, we discuss how to estimate this inner product.


\subsection{Generating the probability vector}

Firstly, let us discuss how to generate $\vec{p}$, a vector which represents $\Tilde{\phi}(t_{\rm ter},\vec{x})$, the probability distribution of $\vec{x}(t_{\rm ter})$, as a quantum state.
As we will see below, although we aim to generate a quantum state in which the amplitudes of basis states are proportional to $\Tilde{\phi}(t_{\rm ter},\vec{x})$, we can apply the method to generate a state in which amplitudes are square roots of probabilities\cite{Grover,Kaneko}, since $\Tilde{\phi}(t_{\rm ter},\vec{x})$ can be regarded as the square roots of the probability densities under another distribution.

Concretely speaking, we aim to generate the vector
\begin{eqnarray}
	\vec{p} &:=& (p_1,...,p_{N_{\rm gr}})^T, \nonumber \\
	p_{k} &:=& \Tilde{\phi}(t_{\rm ter},\vec{x})\prod_{i=1}^d h_i
\label{eq:pVec}
\end{eqnarray}
where $\Tilde{\phi}(t,\vec{x})$, the probability density of $\vec{x}(t)$, is explicitly given as
\begin{eqnarray}
	\Tilde{\phi}(t,\vec{x}) &:=& \frac{1}{(2\pi t)^{d/2}\left(\prod_{i=1}^d\sigma_i\right)\sqrt{\det \rho}}\exp\left(-\frac{1}{2}(\vec{x}-\vec{\mu})^T\Sigma^{-1}(\vec{x}-\vec{\mu})\right), \nonumber \\
	\vec{\mu}&:=&\left(\left(r-\frac{1}{2}\sigma_1^2\right)t,...,\left(r-\frac{1}{2}\sigma_d^2\right)t\right)^T \nonumber \\
	\Sigma&:=&(\sigma_i\sigma_j\rho_{ij})_{\substack{1\le i \le d \\ 1 \le j \le d}},
\end{eqnarray}
that is, the density of the $d$-dimensional normal distribution with the mean $\vec{\mu}$ and the covariance matrix $\Sigma$.
Actually, we generate this vector as a normalized quantum state, that is,
\begin{eqnarray}
	\ket{\bar{p}} &:=&\sum_{k=1}^{N_{\rm gr}} \frac{p_{k}}{P}\ket{k}, \nonumber \\
	P &:=& \|\vec{p}\| =\sqrt{\sum_{k=1}^{N_{\rm gr}} p_{k}^2}.
\end{eqnarray}
Here, note that $(\Tilde{\phi}(t,\vec{x}))^2$ is $\varphi(\vec{x})$ times a constant independent of $\vec{x}$, where
\begin{equation}
	\varphi(\vec{x}) :=  \frac{1}{(\pi t)^{d/2}\left(\prod_{i=1}^d\sigma_i\right)\sqrt{\det \rho}}\exp\left(-\frac{1}{2}(\vec{x}-\vec{\mu})^T\left(\frac{1}{2}\Sigma\right)^{-1}(\vec{x}-\vec{\mu})\right),
\end{equation}
is the probability density function for another $d$-dimensional normal distribution.
Therefore, $\ket{\bar{p}}$ is approximately the state $\ket{\varphi}$, where $\varphi(\vec{x})$ is encoded into the square roots of the amplitudes, that is,
\begin{equation}
	\ket{\varphi} := \frac{1}{\sqrt{Q}}\sum_{k=1}^{N_{\rm gr}} \sqrt{q_{k}}\ket{k}.
\end{equation}
Here,
\begin{equation}
	q_{k} := \int_{x^{(k_1)}_1}^{{x^{(k_1+1)}_1}}dx_1\cdots\int_{x^{(k_d)}_d}^{{x^{(k_d+1)}_d}}dx_d \varphi(t_{\rm ter},\vec{x}), \ {\rm for} \ k=\sum_{i=1}^d n_{\rm gr}^{d-i}k_i+1,k_i=0,1,...,n_{\rm gr}-1,
\end{equation}
which is close to $\varphi(t_{\rm ter},\vec{x}^{(k)})\prod_{i=1}^dh_i$, and
\begin{equation}
	Q := \int_{x^{(0)}_1}^{x^{(n_{\rm gr}+1)}_1}dx_1\cdots\int_{x^{(0)}_d}^{x^{(n_{\rm gr}+1)}_d}dx_d \varphi(t_{\rm ter},\vec{x}),
\end{equation}
which is close to 1.

	\begin{algorithm}[t]
		\caption{Generate $\ket{\bar{p}}$}
		\label{alg1}
		\begin{algorithmic}[1]
			\STATE Prepare $d$ $m_{\rm gr}$-qubit registers and initialize all qubits to $\ket{\bar{0}}$, which means the initial state is $\underbrace{\ket{0}\cdots\ket{0}}_d$.
			\FOR{$i =1$ to $d$}
			\FOR{$j = 1$ to $m_{\rm gr}$}
			\STATE Using $k_1,...,k_{i-1}$ indicated by the first, ..., $(i-1)$-th registers, respectively, and $k_i^{[1]},...,k_i^{[j-1]}$, the bits on the first, ..., $(j-1)$-th qubits of the $i$-th register, respectively,
			rotate the $j$-th qubit in the $i$-th register as \\
			\begin{equation}
			\ket{\bar{0}}\rightarrow \sqrt{f_{i,j}(k_1,...,k_{i-1};k_i^{[1]},...,k_i^{[j-1]})}\ket{\bar{0}}+\sqrt{1-f_{i,j}(k_1,...,k_{i-1};k_i^{[1]},...,k_i^{[j-1]})}\ket{\bar{1}}.
			\end{equation}
			This transforms the entire state into
			\begin{equation}
			\frac{1}{\sqrt{Q}}\sum_{k_1=0}^{n_{\rm gr}-1}\cdots\sum_{k_{i-1}=0}^{n_{\rm gr}-1}\sum_{k_i^{[1]}=0}^{1}\cdots\sum_{k_i^{[j]}=0}^{1}\sqrt{q_{i,j}(k_1,...,k_{i-1};k_i^{[1]},...,k_i^{[j]})}\ket{k_1}\cdots\ket{k_{i-1}}\ket{\Tilde{k}}\underbrace{\ket{0}\cdots\ket{0}}_{d-i},
			\end{equation}
			where $\Tilde{k}$ is an integer whose $m_{\rm gr}$-bit representation is $k_i^{[1]}\cdots k_i^{[j]} \underbrace{0\cdots 0}_{m_{\rm gr}-j}$.
			\ENDFOR
			\ENDFOR
		\end{algorithmic}
	\end{algorithm}

Then, the task is boiled down to generating $\ket{\varphi}$.
This can be done by the multivariate extension of the method of \cite{Grover} for univariate distributions.
The concrete procedure is Algorithm \ref{alg1}.
Here, note that $\ket{k}$ can be decomposed as
\begin{equation}
	\ket{k}=\ket{k_1}\cdots\ket{k_d},
\end{equation}
where each $\ket{k_i}$ is a state on a $m_{\rm gr}$-qubit register (recall that $n_{\rm gr}=2^{m_{\rm gr}}$), and $\ket{k_i}$ can be further decomposed as
\begin{equation}
	\ket{k_i}=\Ket{\overline{k_i^{[i]}}}\cdots\Ket{\overline{k_i^{[m_{\rm gr}]}}},
\end{equation}
where we write the $n$-bit representation of $i\in\{0,1,...,2^n-1\}$ as $i^{[i]}\cdots i^{[n]}$ with $i^{[1]},...,i^{[n]}\in\{0,1\}$.
Besides, note that Algorithm \ref{alg1} requires us to compute
\begin{equation}
	f_{i,j}(k_1,...,k_{i-1};k_i^{[1]},...,k_i^{[j-1]}) := \frac{q_{i,j}(k_1,...,k_{i-1};k_i^{[1]},...,k_i^{[j-1]},0)}{q_{i,j-1}(k_1,...,k_{i-1};k_i^{[1]},...,k_i^{[j-1]})}
\end{equation}
for $i\in[d]$ and $j\in[m_{\rm gr}]$, where
\begin{eqnarray}
	&&q_{i,j}(k_1,...,k_{i-1};b_1,...,b_j) := \nonumber \\
	&& \quad 
	\begin{cases}
		\int_{x^L_{1,j}(b_1,...,b_j)}^{x^R_{1,j}(b_1,...,b_j)}dx_{1}\int_{x^{(0)}_2}^{x^{(n_{\rm gr}+1)}_2}dx_2\cdots\int_{x^{(0)}_d}^{x^{(n_{\rm gr}+1)}_d}dx_d\varphi(t,\vec{x}) & ; \ i=1 \\
		\int_{x^{(k_1)}_1}^{{x^{(k_1+1)}_1}}dx_1\cdots\int_{x^{(k_{i-1})}_{i-1}}^{{x^{(k_{i-1}+1)}_{i-1}}}dx_{i-1} \int_{x^L_{i,j}(b_1,...,b_j)}^{x^R_{i,j}(b_1,...,b_j)}dx_{i}\int_{x^{(0)}_{i+1}}^{x^{(n_{\rm gr}+1)}_{i+1}}dx_{i+1}\cdots\int_{x^{(0)}_d}^{x^{(n_{\rm gr}+1)}_d}dx_d\varphi(t,\vec{x}) & ; \ 2\le i \le d-1 \\
		\int_{x^{(k_1)}_1}^{{x^{(k_1+1)}_1}}dx_1\cdots\int_{x^{(k_{d-1})}_{d-1}}^{{x^{(k_{d-1}+1)}_{d-1}}}dx_{d-1}\int_{x^L_{d,j}(b_1,...,b_j)}^{x^R_{d,j}(b_1,...,b_j)}dx_{d}\varphi(t,\vec{x}) & ; i=d		
	\end{cases}, \nonumber \\
 &&
\end{eqnarray}
and 
\begin{eqnarray}
	&&x^L_{i,j}(b_1,...,b_j):=x^{(k_L)}_i, k_L:=
	\begin{cases}
		0 &; \ j=0 \\
		b_1\cdots b_j \underbrace{0\cdots 0}_{m_{\rm gr}-j} &; \ j=1,...,m_{\rm gr}
	\end{cases}
	\nonumber \\
	&&x^R_{i,j}(b_1,...,b_j):=x^{(k_R)}_i,k_R:=
	\begin{cases}
		n_{\rm gr} &; \ j=0 \\
		b_1\cdots b_j \underbrace{1\cdots 1}_{m_{\rm gr}-j}+1 &; \ j=1,...,m_{\rm gr}
	\end{cases},
	\label{eq:xLxR}
\end{eqnarray}
for $i\in[d]$, $j=0,1,...,m_{\rm gr}$ and $b_1,...,b_j\in\{0,1\}$ (note that $q_{1,0}=Q$).
Such a $f_{i,j}$ can be actually computed as follows.
Neglecting the contribution from the outside of the boundary, we see that
\begin{equation}
	f_{i,j}(k_1,...,k_{i-1};k_i^{[1]},...,k_i^{[j-1]}) \approx \frac{\int_{x^L_{i,j}(b_1,...,b_j)}^{\frac{1}{2}(x^L_{i,j}(b_1,...,b_j)+x^R_{1,j}(b_1,...,b_j))}dx_{i}\varphi^{\rm mar}_i(x_i;k_1,...,k_{i-1})}{\int_{x^L_{i,j}(b_1,...,b_j)}^{x^R_{i,j}(b_1,...,b_j)}dx_{i}\varphi^{\rm mar}_i(x_i;k_1,...,k_{i-1})}, \label{eq:fijApp}
\end{equation}
where
\begin{equation}
	\varphi^{\rm mar}_i(x_i;k_1,...,k_{i-1}):=\int_{-\infty}^{+\infty}dx_{i+1}\cdots\int_{-\infty}^{+\infty}dx_d \varphi((x^{(k_1)}_1,...,x^{(k_{i-1})}_{i-1},x_i,x_{i+1},...,x_d)^T)
\end{equation}
is the marginal density given by integrating out $x_{i+1},...,x_{d}$ and fixing $x_1,...,x_{i-1}$.
We can regard this as an univariate normal distribution density function of $x_i$ (times a constant independent of $x_i$), and therefore compute (\ref{eq:fijApp}) by the method presented in \cite{Kaneko}.

\quad \\

At the end of this subsection, let us evaluate the error of (\ref{eq:pdotYtil}) as an approximation for (\ref{eq:V0ExpPriceApp2}).
As preparation, we evaluate the normalization factor $P$ as follows:
\begin{eqnarray}
	P^2 &=& \sum_{k=1}^{N_{\rm gr}} \left(\phi_{\vec{x}}(t,\vec{x}^{(k)}_{\rm gr})\right)^2\left(\prod_{i=1}^{d}h_i\right)^2 \nonumber \\
	&\approx& \frac{\prod_{i=1}^{d}h_i}{(4\pi t)^{d/2}\left(\prod_{i=1}^d\sigma_i\right)\sqrt{\det\rho}} \int_{\mathbb{R}^d}d\vec{x} \varphi(\vec{x})  \nonumber \\
	&=& \frac{\prod_{i=1}^{d}h_i}{(4\pi t)^{d/2}\left(\prod_{i=1}^d\sigma_i\right)\sqrt{\det\rho}}   \nonumber \\
	&\approx& \frac{\prod_{i=1}^d\Delta_i}{(4\pi)^{d/2}N_{\rm gr}\sqrt{\det\rho}}, \label{eq:P2}
\end{eqnarray}
where
\begin{equation}
	\Delta_i:=\frac{u_i-l_i}{\sigma_i\sqrt{t_{\rm ter}}},i\in[d]. \label{eq:Deltai}
\end{equation}
Besides, we make an additional assumption.
\begin{assum}
	For $Y(\tau,\vec{x})$, the solution of (\ref{eq:PDE2}) and (\ref{eq:BounCond2}), and $\Tilde{\phi}(t,\vec{x})$, the probability density function of $\vec{x}(t)$ under the BS model (\ref{eq:SDE}), there exists $\eta\in\mathbb{R}$ such that
	\begin{equation}
		\forall i,j\in [d], \tau\in(0,T), \vec{x}\in \Tilde{D},
		\left|\frac{\partial^2}{\partial x_i \partial x_j}(\Tilde{\phi}(T-\tau,\vec{x})Y(\tau,\vec{x}))\right| < \eta. \label{eq:phiUDerivCond}
	\end{equation}
	\label{ass:phiYSm}
\end{assum}
\noindent Then, we obtain the following lemma, which guarantees us that we can approximate the integral by the finite sum over the grid points.
\begin{lemma}
	Consider Problem 1.
	Under Assumptions \ref{ass:YSm}, \ref{ass:payoffBound} and \ref{ass:phiYSm}, for a given $\epsilon\in\mathbb{R}_+$ satisfying (\ref{eq:epscond1}) and (\ref{eq:epscond2}), if we set
	\begin{equation}
		h_i < \Tilde{h}_i := \min \left\{\frac{(4\pi)^{d/8}(\det\rho)^{1/8}}{d\sigma_i(\prod_{i=1}^d\Delta_i)^{1/4}}\sqrt{\frac{\epsilon}{2\xi T}},\frac{(4\pi)^{d/8}(\det\rho)^{1/8}}{\sigma_i(\prod_{i=1}^d\Delta_i)^{1/4}}\sqrt{\frac{\epsilon}{\zeta dT}},\frac{1}{\left(\prod_{i=1}^d(u_i-l_i)\right)^{1/2}}\sqrt{\frac{24\epsilon}{d\eta}}\right\}, i\in[d] \label{eq:hxcond3}
	\end{equation}	
	the following holds
	\begin{equation}
		\left|e^{-rT}\vec{p} \cdot \vec{\Tilde{Y}}(\tau_{\rm ter})-V_0\right|<4\epsilon, \label{eq:finSumUB}
	\end{equation}
	where $\vec{p}$ is defined as (\ref{eq:pVec}), $\vec{\Tilde{Y}}$ is the solution of (\ref{eq:ODE}).
	\label{lem:finSumApp}
\end{lemma}
\noindent The proof is given in Appendix \ref{sec:PrLemFinSumApp}.

\subsection{Generating the derivative price vector}

Next, let us consider how to generate $\vec{V}$, the vector which encodes the grid derivative prices at $t_{\rm ter}$.
Precisely speaking, since we solve (\ref{eq:ODE}), we actually obtain the vector $\vec{\Tilde{Y}}$, which encodes the approximations of $Y(\tau_{\rm ter}, \vec{x})$ on the grid points.
Furthermore, by the algorithm presented in Section \ref{sec:QAlgODE}, we obtain not $\vec{\Tilde{Y}}$ itself but some quantum state like (\ref{eq:outputBerryAlgo}), which contains a state corresponding to $\vec{\Tilde{Y}}$ along with a garbage state.

For the precise discussion, let us firstly make some assumptions in order to satisfy preconditions to use the quantum algorithm.
The first one is as follows:
\begin{assum}
	$\vec{C}(\tau)$ in (\ref{eq:ODE}) is independent of $\tau$. \label{ass:constC}
\end{assum}
\noindent Then, hereafter, we simply write $\vec{C}(\tau)$ as $\vec{C}$.
We make this assumption in order to fit the current setting to \cite{Berry2}, which considered solving (\ref{eq:ODEEx}) for constant $A$ and $\vec{b}$ (note that $F$ in (\ref{eq:ODE}) is constant).
Although $\vec{C}(\tau)$ is not generally time-independent, the assumption is satisfied in some cases:
\begin{itemize}
	\item In some cases, a derivative is far in-the-money for a party at some points on the boundary, and this means that the party would receive a constant payoff $K$ at the maturity with high probability.
	For example,
	\begin{itemize}
	\item The payoff is the cash-or-nothing type.
	\item The payoff is capped, that is, the payoff function takes the form of $f_{\rm pay}(\vec{S})=\min\{f(\vec{S}),K\}$ with some function $f(\vec{S})$.
	\end{itemize}
 	In these cases, we can approximate that $V(t,\vec{S})\approx e^{-r(T-t)}K$, which means that $Y(\tau,\vec{x})\approx K$, on the points.
 	
 	\item If a boundary corresponds to a knock-out barrier, $V(t,\vec{S})=0$ on it.
\end{itemize}
Of course, there are many cases where $\vec{C}(\tau)$ is time-dependent, and it is desirable to expend our method to such cases.
We leave this as a future work.

The second assumption is as follows:
\begin{assum}
	For $F$ in (\ref{eq:F}), the following oracles $O_{F,1}$ and $O_{F,2}$ are available:
	\begin{equation}
		O_{F,1}:\ket{j}\ket{l} \mapsto \ket{j}\ket{\nu_F(j,l)},
	\end{equation}
	where $j\in [N_{\rm gr}]$, $l\in [s_F]$, $s_F$ is the sparsity of $F$, and $\nu(j,l)$ is the column index of the $l$-th nonzero entry in the $j$-th row,
	\begin{equation}
		O_{F,2}:\ket{j}\ket{k}\ket{z} \mapsto \ket{j}\ket{k}\ket{z\oplus F_{jk}},
	\end{equation}
	where $j,k\in [N_{\rm gr}]$ and $z\in \mathbb{R}$.
	Besides, for $\vec{\Tilde{f}}_{\rm pay}$ in (\ref{eq:IniODE}) and $\vec{C}$ in (\ref{eq:C}), we know their norms and the following oracles $O_{\vec{\Tilde{f}}_{\rm pay}}$ and $O_{\vec{C}}$ are available:
	\begin{equation}
		O_{\vec{\Tilde{f}}_{\rm pay}}:
		\begin{cases}
			\ket{\bar{0}}\ket{0} \mapsto \frac{1}{\|\vec{\Tilde{f}}_{\rm pay}\|}\ket{\bar{0}}\ket{\vec{\Tilde{f}}_{\rm pay}} \\
			\ket{\bar{1}}\ket{\psi} \mapsto \ket{\bar{1}}\ket{\psi} \ {\rm for \ any} \ \ket{\psi}
		\end{cases},
	\end{equation}
	\begin{equation}
		O_{\vec{C}}:
		\begin{cases}
			\ket{\bar{0}}\ket{\psi} \mapsto \ket{\bar{0}}\ket{\psi} \ {\rm for \ any} \ \ket{\psi} \\
			\ket{\bar{1}}\ket{0} \mapsto \frac{1}{\|\vec{C}\|}\ket{\bar{1}}\ket{\vec{C}}
		\end{cases},
	\end{equation}
	for $\vec{C}\ne 0$ and $O_{\vec{C}}$ is an identity operator for $\vec{C}= 0$.
	\label{ass:oracle}
\end{assum}
\noindent Since $F$ is explicitly given as (\ref{eq:F}), the sum of the Kronecker products of tridiagonal matrices, construction of $O_{F,1}$ and $O_{F,2}$ is straightforward.
On the other hand, $\vec{\Tilde{f}}_{\rm pay}$ and $\vec{C}$ are highly problem-dependent, and so are $O_{\vec{\Tilde{f}}_{\rm pay}}$ and $O_{\vec{C}}$.
Therefore, we just assume their availability in this paper, referring to some specific cases.

\begin{itemize}
	\item By the analogy with preparation of $\ket{\bar{p}}$, we see that we can prepare $\frac{\ket{\vec{\Tilde{f}}_{\rm pay}}}{\|\vec{\Tilde{f}}_{\rm pay}\|}$ if we can efficiently calculate
	\begin{equation}
			\frac{\int_{x^L_{i,j}(b_1,\cdots,b_j)}^{\frac{1}{2}(x^L_{i,j}(b_1,\cdots,b_j)+x^R_{i,j}(b_1,\cdots,b_j))}dx_{i}\int_{l_i}^{u_i}dx_{i+1}\cdots\int_{l_d}^{u_d}dx_d\left(\Tilde{f}_{\rm pay}((x^{(k_1)}_1,\cdots,x^{(k_{i-1})}_{i-1},x_i,x_{i+1},\cdots,x_d)^T)\right)^2}{\int_{x^L_{i,j}(b_1,\cdots,b_j)}^{x^R_{i,j}(b_1,\cdots,b_j)}dx_{i}\int_{l_i}^{u_i}dx_{i+1}\cdots\int_{l_d}^{u_d}dx_d\left(\Tilde{f}_{\rm pay}((x^{(k_1)}_1,\cdots,x^{(k_{i-1})}_{i-1},x_i,x_{i+1},\cdots,x_d)^T)\right)^2},
	\end{equation}
	where $x^L_{i,j}$ and $x^R_{i,j}$ are defined as (\ref{eq:xLxR}), for $i\in[d]$,$j\in\{0,1,\cdots,m_{\rm gr}\}$,$k_1,\cdots,k_d\in\{0,1,\cdots,n_{\rm gr}-1\}$ and $b_1,\cdots,b_j\in\{0,1\}$.
	Although it is difficult to analytically calculate this in general, there are some cases where it is possible.
	An example is the case where $f_{\rm pay}$ depends on only one underlying asset price, say $S_1$ (and other assets are relevant to the barrier), and has a simple function form, e.g. $f_{\rm pay}(\vec{S})=\max\{S_1-K,0\}$.
	
	\item If all boundaries correspond to knock-out barriers, $\vec{C}=\vec{0}$, and therefore $O_{\vec{C}}$ is just an identity operator.
\end{itemize}

Then, we obtain the following lemma, whose proof is presented in Appendix \ref{sec:PrCompOnePsi}.

\begin{lemma}
	Consider the ODE system (\ref{eq:ODE}).
	Assume that Assumptions \ref{ass:YSm}, \ref{ass:payoffBound}, \ref{ass:phiYSm}, \ref{ass:constC} and \ref{ass:oracle} are satisfied.
	Let $\epsilon$ be any positive real number satisfying (\ref{eq:epscond1}) and (\ref{eq:epscond2}), and $\epsilon^\prime$ be any positive real number.
	Then, there exists a quantum algorithm that produces a state $\ket{\Tilde{\Psi}}$ $\epsilon^\prime$-close to
	\begin{equation}
		\ket{\Psi} := \frac{1}{\sqrt{\braket{\Psi_{\rm gar}|\Psi_{\rm gar}}+(p+1)\|\vec{\Tilde{Y}}(\tau_{\rm ter})\|^2}}\left(\ket{\Psi_{\rm gar}}+\sum_{j=p(k+1)}^{p(k+2)}\ket{j}\ket{\vec{\Tilde{Y}}(\tau_{\rm ter})}\right), \label{eq:outputBerryAlgoFDM}
	\end{equation}
	where $\vec{\Tilde{Y}}(\tau_{\rm ter})$ is a vector satisfying (\ref{eq:finSumUB}), using
	\begin{equation}
		O\left(\mathcal{C}\times {\rm poly}\left(\log \left(\frac{\mathcal{C}}{\epsilon^\prime}\right)\right)\right) 
		 \label{eq:compQODEFDM}
	\end{equation}
	queries to $O_{F,1}$, $O_{F,2}$, $O_{\vec{\Tilde{f}}_{\rm pay}}$, and $O_{\vec{C}}$.
	Here,
	\begin{equation}
		\mathcal{C} := \max \left\{\frac{\sqrt{\prod_{i=1}^d \Delta_i}d^2 \Xi \sigma^2_{\rm max}\tau_{\rm ter}}{(4\pi)^{d/4}(\det \rho)^{1/4}},d\eta\prod_{i=1}^d(u_i-l_i)\right\}\times \frac{\kappa_Vd^4\sigma^2_{\rm max}\tau_{\rm ter}}{\epsilon},
	\end{equation}
	$\kappa_V = \|V\|\cdot\|V^{-1}\|$ is the condition number of $V$ which diagonalizes $F$ (i.e. $V FV^{-1}$ is a diagonal matrix), $\sigma_{\rm max}:=\max_{i\in[d]}\sigma_i$, $\Xi:=\max\{\xi,\zeta/d\}$, $\tau_{\rm ter}:=T-t_{\rm ter}$, $t_{\rm ter}$ is defined as (\ref{eq:tter}), $p:=\lceil \tau_{\rm ter}\|F\| \rceil$, $k:=\lfloor 2\log\Omega/\log(\log\Omega)\rfloor$, $\Omega=70g\kappa_Vp^{3/2}(\|\vec{f}_{\rm pay}\|+T\|\vec{C}\|)/\epsilon\|\vec{\Tilde{Y}}(\tau_{\rm ter})\|$, and $\ket{\Psi_{\rm gar}}$ is an unnormalized state which takes the form of $\ket{\Psi_{\rm gar}}=\sum_{j=0}^{p(k+1)-1}\ket{j}\ket{\psi_j}$ with some unnormalized states $\ket{\psi_0},\ket{\psi_1},...,\ket{\psi_{p(k+1)-1}}$ and satisfies
	\begin{equation}
	\braket{\Psi_{\rm gar}|\Psi_{\rm gar}}=O(g^2(p+1)\|\vec{\Tilde{Y}}(\tau_{\rm ter})\|^2) \label{eq:PsiTrashNorm}
	\end{equation}
	with $g := \max_{\tau\in[0,\tau_{\rm ter}]}\|\vec{\Tilde{Y}}(\tau)\|/\|\vec{\Tilde{Y}}(\tau_{\rm ter})\|$.
	\label{lem:CompOnePsi}
\end{lemma}

\subsection{Proposed algorithm}

Finally, based on the above discussions, we present the quantum method to calculate the present derivative price $V_0$.
Our strategy is calculating this as (\ref{eq:pdotYtil}).
More concretely, we aim to subtract the information of $\vec{p}\cdot\vec{\Tilde{Y}}(\tau_{\rm ter})$ from $\ket{\Psi}$ in (\ref{eq:outputBerryAlgoFDM}), the output state of the algorithm of \cite{Berry2}. 

In order to do this, we first modify the algorithm slightly.
That is, we aim to solve not (\ref{eq:LSBerry}) but the following one by the QLS algorithm:
\begin{equation}
	\Tilde{C}_{m,k,p}(Fh_t)\vec{X}=\vec{e}_0\otimes\vec{\Tilde{f}}_{\rm pay}+h_t\sum_{i=0}^{m-1}\vec{e}_{i(k+1)+1}\otimes\vec{C}+\sum_{i=1}^{p+1}\vec{e}_{m(k+1)+p+i}\otimes\vec{\gamma}. \label{eq:LSMod}
\end{equation}
Here, $m,p,k$ are integers defined in the statement of Lemma \ref{lem:CompOnePsi}, $q:=m(k+1)+2p+1$, $h_t=\tau_{\rm ter}/m$, $\vec{X}\in\mathbb{R}^{N_{\rm gr}(q+1)}$, $\{\vec{e}_i\}_{i=0,1,...,q}$ is an orthonormal basis of $\mathbb{R}^{q+1}$, and $\vec{\gamma}:=(\gamma,...,\gamma)^T\in \mathbb{R}^{N_{\rm gr}}$ for some $\gamma \in \mathbb{R}_+$.
Hereafter, we make the following assumption on $\gamma$:
\begin{assum}
	We are given $\gamma\in\mathbb{R}_+$ satisfying
	\begin{equation}
		\frac{1}{2} \bar{Y}(\tau_{\rm ter}) < \gamma < 2\bar{Y}(\tau_{\rm ter}), \label{eq:gammacond}
	\end{equation}
	where
	\begin{equation}
	\bar{Y}(\tau_{\rm ter}) := \sqrt{\frac{1}{N_{\rm gr}}\sum_{k=1}^{N_{\rm gr}}(Y(\tau_{\rm ter},\vec{x}^{(k)}))^2}.
	\end{equation}
	\label{ass:gamma}
\end{assum}
\noindent This means that $\gamma$ is comparable with the root mean square of $Y(\tau_{\rm ter},\vec{x})$ on the grid points.
Besides, the $N_{\rm gr}(q+1)\times N_{\rm gr}(q+1)$ matrix $\Tilde{C}_{m,k,p}(Fh_t)$ is now defined as
\begin{eqnarray}
	\Tilde{C}_{m,k,p}(Fh_t)&:=&\sum_{j=0}^{q}\vec{e}_j\vec{e}_j^T\otimes I_{N_{\rm gr}}-\sum_{i=0}^{m-1}\sum_{j=1}^k\vec{e}_{i(k+1)+j}\vec{e}_{i(k+1)+j-1}^T\otimes \frac{1}{j}Fh_t \nonumber \\
	&& \quad -\sum_{i=0}^{m-1}\sum_{j=0}^k\vec{e}_{(i+1)(k+1)}\vec{e}_{i(k+1)+j}^T\otimes I_{N_{\rm gr}} - \sum_{j=m(k+1)+1}^{m(k+1)+p}\vec{e}_j\vec{e}_{j-1}^T\otimes I_{N_{\rm gr}},
\end{eqnarray}
or, equivalently,
\begin{equation}
	\Tilde{C}_{m,k,p}(Fh_t) =
	\begin{pmatrix}
		C_{m,k,p}(Fh_t) & 0 \\
		0 & I_{N_{\rm gr}(p+1)}
	\end{pmatrix}.
\end{equation}
Visually, (\ref{eq:LSMod}) is displayed as follows
\small
\begin{equation}
	\begin{pmatrix}
		I_{N_{\rm gr}}       &        &         &           & & & &  & & & &&&& \\
		-Fh_t/1 & I_{N_{\rm gr}}      &         &           &&&&&&&&&&& \\
		& \ddots & \ddots  &           &&&&&&&&&&& \\
		&        & -Fh_t/k & I_{N_{\rm gr}}         &&&&&&&&&&& \\
		-I_{N_{\rm gr}}      & \cdots & -I_{N_{\rm gr}}      & -I_{N_{\rm gr}}     & I_{N_{\rm gr}}      &         &&&&&&&&&\\
		&        &         &        & \ddots & \ddots  & &&&&&&&&\\
		&        &         &        &        & -Fh_t/1 & I_{N_{\rm gr}} &&&&&&&&\\
		&		 &         &        &        &         & \ddots & \ddots  & &&& &&&\\
		&		 &         &        &        &         &        & -Fh_t/k & I_{N_{\rm gr}}      &&& &&&\\
		&		 &         &        &        & -I_{N_{\rm gr}}      & \cdots & -I_{N_{\rm gr}}      & -I_{N_{\rm gr}}     & I_{N_{\rm gr}}       & & &&&\\
		&		 &         &        &        &         &        &         &        & -I_{N_{\rm gr}}      & I_{N_{\rm gr}}      & &&&\\
		&&		 &         &        &        &         &        &         &        &          \ddots & \ddots & &&&\\
		&		 &         &       & &        &         &        &         &        &          & -I_{N_{\rm gr}} & I_{N_{\rm gr}} &&&\\
		&		 &         &       & &        &         &        &         &        &          &  &  & I_{N_{\rm gr}} & & \\
		&		 &         &       & &        &         &        &         &        &          &  &  & & \ddots & \\
		&		 &         &       & &        &         &        &         &        &          &  &  & & & I_{N_{\rm gr}} \\
	\end{pmatrix}
	\vec{X}
	=
	\begin{pmatrix}
		\vec{\Tilde{f}}_{\rm pay} \\
		h_t\vec{C}\\
		0\\
		\vdots\\
		0\\
		\vdots\\
		h_t\vec{C}\\
		0\\
		\vdots\\
		0\\
		0\\
		\vdots\\
		0\\
		\vec{\gamma}\\
		\vdots\\
		\vec{\gamma}
	\end{pmatrix}.
\end{equation}
\normalsize
The solution of (\ref{eq:LSMod}) is
\begin{equation}
	\vec{X}=\sum_{i=0}^{m-1}\sum_{j=1}^k\vec{e}_{i(k+1)+j}\otimes\vec{\Tilde{\Tilde{Y}}}_{i,j} + \sum_{j=0}^p\vec{e}_{m(k+1)+j}\otimes\vec{\Tilde{\Tilde{Y}}}(\tau_{\rm ter}) + \sum_{j=1}^{p+1}\vec{e}_{m(k+1)+p+j}\otimes\vec{\gamma},
\end{equation}
for some vectors $\vec{\Tilde{\Tilde{Y}}}_{i,j},\vec{\Tilde{\Tilde{Y}}}(\tau_{\rm ter})\in \mathbb{R}^{N_{\rm gr}}$, and $\vec{\Tilde{\Tilde{Y}}}(\tau_{\rm ter})$ becomes close to $\vec{\Tilde{Y}}(\tau_{\rm ter})$.
Note that, in $\vec{X}$, $\vec{\Tilde{\Tilde{Y}}}(\tau_{\rm ter})$ and $\vec{\gamma}$ are repeated $(p+1)$-times.
Then, applying the quantum algorithm, we can generate the quantum state $\ket{\Tilde{\Psi}_{\rm mod}}$ $\epsilon$-close to
\begin{eqnarray}
	\ket{\Psi_{\rm mod}} &:=& \frac{1}{Z}\left(\ket{\Psi_{\rm gar}}+\sum_{j=p(k+1)}^{p(k+2)}\ket{j}\ket{\vec{\Tilde{Y}}(\tau_{\rm ter})}+\sum_{j=p(k+2)+1}^{p(k+3)+1}\ket{j}\ket{\vec{\gamma}}\right), \nonumber \\
	Z&:=&\sqrt{\braket{\Psi_{\rm gar}|\Psi_{\rm gar}}+(p+1)\|\vec{\Tilde{Y}}(\tau_{\rm ter})\|^2+(p+1)N_{\rm gr}\gamma^2}.
	 \label{eq:outputBerryAlgoFDMMod}
\end{eqnarray}
Note that the query complexity for generating $\ket{\Tilde{\Psi}_{\rm mod}}$ is (\ref{eq:compQODEFDM}), similarly to $\ket{\Tilde{\Psi}}$.
This is because the complexity of the QLS algorithm depends only on the condition number and sparsity of the matrix and the tolerance\cite{Childs3}, and the condition number and sparsity of $\Tilde{C}_{m,k,p}(Fh_t)$ is same as $C_{m,k,p}(Fh_t)$.

\begin{algorithm}[t]
	\caption{Calculate $e^{-rT}\vec{p}\cdot\vec{\Tilde{Y}}(\tau_{\rm ter})$}
	\label{alg2}
	\begin{algorithmic}[1]
		\REQUIRE{\ \\
			$\gamma\in\mathbb{R}_+$ satisfying (\ref{eq:gammacond}).\\
			$\epsilon\in\mathbb{R}_+$ satisfying (\ref{eq:epscond1}) and (\ref{eq:epscond2}).\\
			$\epsilon_1\in\mathbb{R}_+$ satisfying (\ref{eq:eps1cond}).\\
			$\epsilon_2\in\mathbb{R}_+$ satisfying (\ref{eq:eps2cond}).\\ 
			$\epsilon_{\Tilde{\Psi}_{\rm mod}}\in\mathbb{R}_+$ satisfying (\ref{eq:epsPsiCond}). \\
			Accesses to the oracle $U_{\Tilde{\Psi}_{\rm mod}}$ such that (\ref{eq:UPsiMod}) and (\ref{eq:tilPsiCond}) and its inverse. \\
			Accesses to the oracle $U_{\Pi}$ such that (\ref{eq:Upi}) and its inverse.	}
		\STATE Estimate the amplitude of $\ket{0}\ket{0}$ in the state $U_{\Pi}^\dagger U_{\Tilde{\Psi}_{\rm mod}}\ket{0}\ket{0}$ by QAE with tolerance $\epsilon_1$. Let the output be $E_1$.
		\STATE Estimate the square root of the probability that we obtain either of $p(k+2)+1,...,p(k+3)+1$ when we measure the first register of $\ket{\Tilde{\Psi}_{\rm mod}}$ by QAE with tolerance $\epsilon_2$. Let the output be $E_2$.
		\STATE Output $e^{-rT}\frac{\gamma\sqrt{N_{\rm gr}}PE_1}{E_2}=:\omega$, where $P$ is given by (\ref{eq:P2}).
	\end{algorithmic}
\end{algorithm}

Using $\ket{\Tilde{\Psi}_{\rm mod}}$, we can estimate $\vec{p}\cdot\vec{\Tilde{Y}}(\tau_{\rm ter})$.
The outline is as follows.
First, we estimate the inner product
\begin{equation}
	\braket{\Pi | \Psi_{\rm mod}} = \frac{\sqrt{p+1}}{PZ}\vec{p}\cdot\vec{\Tilde{Y}}(\tau_{\rm ter}),
\end{equation}
where
\begin{equation}
	\ket{\Pi}:=\frac{1}{\sqrt{p+1}}\sum_{j=p(k+1)}^{p(k+2)}\ket{j}\ket{\bar{p}},
\end{equation}
by estimating the amplitude of $\ket{0}\ket{0}$ in $U_{\Pi}^\dagger U_{\Tilde{\Psi}_{\rm mod},\epsilon}\ket{0}\ket{0}$ using QAE.
Here, $U_{\Tilde{\Psi}_{\rm mod}}$ and $U_{\Pi}$ are the unitary operators such that
\begin{equation}
	U_{\Tilde{\Psi}_{\rm mod}}\ket{0}\ket{0}=\ket{\Tilde{\Psi}_{\rm mod}}, \label{eq:UPsiMod}
\end{equation}
and
\begin{equation}
U_{\Pi}\ket{0}\ket{0}=\ket{\Pi}, \label{eq:Upi}
\end{equation}
respectively.
Note that, if we can generate $\ket{\bar{p}}$, we can also generate $\ket{\Pi}$, since this is just a tensor product of $\frac{1}{\sqrt{p+1}}\sum_{j=p(k+1)}^{p(k+2)}\ket{j}$ and $\ket{\bar{p}}$.
Next, by QAE, we estimate the probability that we obtain $j\in\{p(k+2)+1,...,p(k+3)+1\}$ in the first register when we measure $\ket{\Tilde{\Psi}_{\rm mod}}$, and then obtain an estimation of $\gamma\sqrt{(p+1)N_{\rm gr}}/Z$.
Finally, using $E_1$ and $E_2$, the outputs of the first and second estimations, respectively, we calculate
\begin{equation}
	\frac{e^{-rT}\gamma\sqrt{N_{\rm gr}}PE_1}{E_2}
\end{equation}
as an estimation of $e^{-rT}\vec{p}\cdot\vec{\Tilde{Y}}(\tau_{\rm ter})$.
We present the detailed procedure is described as Algorithm \ref{alg2}.
Here, taking some $\epsilon\in \mathbb{R}_+$, we require the tolerances $\epsilon_1$ and $\epsilon_2$ in calculating $E_1$ and $E_2$ be
\begin{eqnarray}
	\epsilon_1 &=& O\left(\frac{(2\pi)^{d/2}\sqrt{\det \rho}\epsilon}{g\left(\prod_{i=1}^d\Delta_i\right)\bar{V}}\right), \label{eq:eps1cond} \\
	\epsilon_2 &=& O\left(\frac{\epsilon}{gV_0}\right) \label{eq:eps2cond}
\end{eqnarray}
respectively, where
\begin{eqnarray}
	\bar{V}(t_{\rm ter}) &:=& \sqrt{\frac{1}{N_{\rm gr}}\sum_{k=1}^{N_{\rm gr}}(V(t_{\rm ter},\vec{S}^{(k)}))^2}, \nonumber \\
	\vec{S}^{(k)}&:=&(S^{(k)}_1,...,S^{(k)}_d)^T:=(\exp({x^{(k_1)}_1}),...,\exp({x^{(k_d)}_d}))^T \ {\rm for} \ k=\sum_{i=1}^d n_{\rm gr}^{d-i}k_i+1,k_i=0,1,...,n_{\rm gr}-1  \nonumber \\
	\quad
\end{eqnarray}
is the root mean square of the derivative prices on the grid points at time $t_{\rm ter}$.
Besides, we require that
\begin{equation}
 \|\ket{\Tilde{\Psi}_{\rm mod}}-\ket{\Psi_{\rm mod}}\|<\epsilon_{\Psi}, \label{eq:tilPsiCond}
\end{equation}
where
\begin{equation}
	\epsilon_{\Psi} = O\left(\max\{\epsilon_1,\epsilon_2\}\right). \label{eq:epsPsiCond}
\end{equation}
These requirements guarantee the overall error to be smaller than $\epsilon$.
We formally state these points along with the complexity of the procedure in Theorem \ref{th:main}, whose proof is presented in Appendix \ref{sec:PrThMain}.

\begin{theorem}
	Consider Problem 1.
	Assume that Assumptions \ref{ass:YSm}, \ref{ass:payoffBound}, \ref{ass:phiYSm}, \ref{ass:constC}, \ref{ass:oracle} and \ref{ass:gamma} are satisfied.
	Then, for any $\epsilon\in\mathbb{R}_+$ satisfying (\ref{eq:epscond1}) and (\ref{eq:epscond2}), Algorithm \ref{alg2} outputs the real number $\omega$ such that
	\begin{equation}
		|\omega-V_0|=O(\epsilon) \label{eq:finalErr}
	\end{equation}
	with
	\begin{equation}
		O\left(\mathcal{D}\times {\rm poly}\left(\log \mathcal{D}\right)\right) 
		\label{eq:compFinal}
	\end{equation}
	queries to $O_{F,1}$, $O_{F,2}$, $O_{\vec{\Tilde{f}}_{\rm pay}}$, and $O_{\vec{C}}$, where
	\begin{equation}
		\mathcal{D} := \max \left\{\frac{\sqrt{\prod_{i=1}^d \Delta_i}d^2 \Xi \sigma^2_{\rm max}\tau_{\rm ter}}{(4\pi)^{d/4}(\det \rho)^{1/4}},d\eta\prod_{i=1}^d(u_i-l_i)\right\}\times\max\left\{\frac{\left(\prod_{i=1}^d\Delta_i\right)\bar{V}}{(2\pi)^{d/2}\sqrt{\det \rho}},V_0\right\}\times \frac{g\kappa_Vd^4\sigma^2_{\rm max}\tau_{\rm ter}}{\epsilon^2},
	\end{equation}
	$\sigma_{\rm max}:=\max_{i\in[d]}\sigma_i$, $\Xi:=\max\{\xi,\zeta/d\}$, $\tau_{\rm ter}:=T-t_{\rm ter}$, $t_{\rm ter}$ is defined as (\ref{eq:tter}), $\Delta_i$ is defined as (\ref{eq:Deltai}), $g := \max_{\tau\in[0,\tau_{\rm ter}]}	\|\vec{\Tilde{Y}}(\tau)\|/\|\vec{\Tilde{Y}}(\tau_{\rm ter})\|$, $\bar{V}(t_{\rm ter}) := \sqrt{\frac{1}{N_{\rm gr}}\sum_{k=1}^{N_{\rm gr}}(V(t_{\rm ter},\vec{S}^{(k)}))^2}$, and $\kappa_V = \|V\|\cdot\|V^{-1}\|$ is the condition number of $V$, which diagonalizes $F$.
	\label{th:main}
\end{theorem}

Let us make some comments.
First, note that the upper bound of the complexity (\ref{eq:compFinal}) does not have any factor like $(1/\epsilon)^{{\rm poly}(d)}$, which means the tremendous speedup with respect to $\epsilon$ and $d$ compared with the classical FDM.
On the other hand, the exponential dependence on $d$ has not completely disappeared.
In fact, (\ref{eq:compFinal}) contains some constants to the power of $d$, and factors such as $\prod_{i=1}^d(u_i-l_i)$ and $\prod_{i=1}^d\Delta_i$, that is, the $d$-times product of $u_i-l_i$ or $\Delta_i$.
Recall that $u_i-l_i=\log (U_i/L_i)$ is the width between boundaries in the direction of $x_i$, the logarithm of the $i$-th underlying asset price, and $\Delta_i$ is that divided by $\sigma_i \sqrt{t_{\rm ter}}$, which roughly measures the extent of the probability distribution of $x_i$ at time $t_{\rm ter}$.
Therefore, these factors are just logarithmic factors to the power of $d$.

Second, we note that some calculation parameters are difficult to be determined in advance of pricing.
For example, although we have assumed that we know $\gamma$ such that (\ref{eq:gammacond}) holds in advance, it is difficult because we do not know $\bar{Y}(\tau_{\rm ter})$.
Besides, although we set $p=\lceil \|F\|\tau_{\rm ter}\rceil$ in using the algorithm of \cite{Berry2}, it is difficult to set $p$ to this specific value since we can upper bound $\|F\|$ but cannot calculate it precisely.
Even in \cite{Berry2}, the way to set $p=\lceil \|F\|\tau_{\rm ter}\rceil$ is not presented.
Similar discussion can be applied to other parameters: $k$, $h_i$, and so on. 
Fortunately, the algorithm works not only for such specific values of the parameters but also for comparable values.
The factor $1/2$ and $2$ in (\ref{eq:gammacond}) can be replaced with comparable values (say, $1/3$ and $3$), which results in change of the complexity only by some $O(1)$ factor.
$p$ larger than but comparable with $\lceil \|F\|\tau_{\rm ter}\rceil$ (say, $2\lceil \|F\|\tau_{\rm ter}\rceil$) results in comparable computational accuracy and complexity with those for $p=\lceil \|F\|\tau_{\rm ter}\rceil$.
In reality, we may perform computation for various parameter values and search the appropriate ranges of the parameters, for which the calculated derivative price seems to converge.
In the practical business, once we find a set of appropriate calculation parameters, we can continue to use it with periodic check of convergence, since we typically perform pricing many times in different but similar settings on model parameters (e.g. $\sigma_i$) and contract terms (e.g. barrier level).

\section{Summary\label{sec:Sum}}

In this paper, we studied how to apply the quantum algorithm of \cite{Berry2} for solving linear differential equations to pricing multi-asset derivatives by FDM.
As we explained, FDM is an appropriate method for pricing some types of derivatives such as barrier options, but suffers from the so-called curse of dimensionality, which makes FDM infeasible for large $d$, the number of underlying assets, since the dimension of the corresponding ODE system grows as $\epsilon^{{\rm poly}(d)}$ for the tolerance $\epsilon$, and so does the complexity.
We saw that the quantum algorithm for solving ODE systems, which provides the exponential speedup with respect to the dimensionality compared with classical methods, is beneficial also for derivative pricing.
In order to address the specific issue for derivative pricing, that is, extracting the present price from the output state of the quantum algorithm, we adopted the strategy that we calculate the present price as the expected value of the price at some appropriate future time $t_{\rm ter}$.
Then, we constructed the concrete calculation procedure, which is combination of the algorithm of \cite{Berry2} and QAE.
We also estimated the query complexity of our method, which does not have any dependence like $(1/\epsilon)^{{\rm poly}(d)}$ and shows tremendous speedup with respect to $\epsilon$ and $d$.

We believe that this paper is the first step for the research in this direction, but there remains many points to be improved.
First, we should consider whether the assumptions we made can be mitigated.
For example, although we assume that $\vec{C}(\tau)$ is time-independent (Assumption \ref{ass:constC}), some products do not fit to this condition: e.g., when we consider the upper boundary condition in the case of the European-call-like payoff $f_{\rm pay}(S)=\max\{S-K,0\}$ with some constant $K$, $V(t,S)\approx S-e^{-r(T-t)}K$ and therefore $Y(\tau,\vec{x})=e^{r\tau}V(t,\vec{S})$ cannot be regarded as constant for large $S$.
In order to omit this assumption, we might be able to extend the algorithm of \cite{Berry2} so that it can be applied to time-dependent $\vec{C}(\tau)$\footnote{Actually, the algorithm in \cite{Childs}, which is based on the spectral method, can deal with time-dependent $\vec{C}(\tau)$. However, in order to apply this algorithm, $V(t,\vec{S})$ must be smooth enough in the direction of $t$. On the other hand, in practice, the BS model parameters are often not smooth: for example, piece-wise constant volatilities are often used, which deteriorates smoothness of $V(t,\vec{S})$. In such a case, the algorithm of \cite{Berry2} is expected to be more suitable than that of \cite{Childs}, since the formal solution (\ref{eq:formSol}) is valid also for piece-wise constant model parameters. That is, if
\begin{equation}
	A=
	\begin{cases}
		A_1 & ; \ 0 \le t \le t_{\rm dis} \\
		A_2 & ; \ t_{\rm dis} < t \le T \\
	\end{cases},
	\vec{b}=
	\begin{cases}
		\vec{b}_1 & ; \ 0 \le t \le t_{\rm dis} \\
		\vec{b}_2 & ; \ t_{\rm dis} < t \le T \\
	\end{cases}, \nonumber
\end{equation}	
with some $t_{\rm dis}\in (0,T)$, $\vec{x}(t)$ can be written as
\begin{equation}
	\vec{x}(t)=
	\begin{cases}
		e^{A_1t}\vec{x}_{\rm ini}+(e^{A_1t}-I_N)A_1^{-1}\vec{b}_1 & ; \ 0 \le t \le t_{\rm dis} \\
		e^{A_2(t-t_{\rm dis})}\vec{x}(t_{\rm dis})+(e^{A_2(t-t_{\rm dis})}-I_N)A_2^{-1}\vec{b}_2 & ; \ t_{\rm dis} < t \le T \\
	\end{cases}. \nonumber
\end{equation}
}

Another important aspect is pricing early-exercisable derivatives.
American-type (resp. Bermudan-type) derivatives, in which either of parties can terminate the contract at any time (resp. at either of some predetermined dates) before the final maturity $T$, are widely traded and their pricing is important for banks.
FDM is suitable and often used for pricing such products, since it determines the derivative price backward from $T$ and can take into account early exercise.
However, it is not straightforward to apply the quantum method proposed in this paper to pricing early-exercisable products.
This is because, at exercisable date $t_{\rm exe}$, we need the operation $V(t_{\rm exe},\vec{S})=\max\{V(t_{{\rm exe}}+0,\vec{S}), f_{\rm pay}(\vec{S})\}$, where $V(t_{{\rm exe}}+0,\vec{S})$ is the derivative price right after $t_{\rm exe}$, but nonlinear operations on amplitudes such as the max function cannot be implemented on a quantum computer naively.

Including these points, we will investigate the possibility that the quantum FDM speedups pricing for the wider range of derivatives in the future work.

\section*{Acknowledgment}

This work was supported by MEXT Quantum Leap Flagship Program (MEXT Q-LEAP) Grant Number JPMXS0120319794.

\appendix

\section{Proofs\label{sec:Pr}}


\subsection{Proof of Lemma \ref{lem:FDMErr} \label{sec:PrLemFDMErr}}

First, we prove the following property of $F$ in (\ref{eq:F}).

\begin{lemma}
	For $F$ in (\ref{eq:F}), the logarithmic norm satisfies $\mu(F)<0$.
\end{lemma}

\begin{proof}
	Since the matrix $(\rho_{ij})_{1\le i,j\le d}$ is positive-definite, so is the matrix $\left(\sigma_i\sigma_j\rho_{ij}\right)_{1\le i,j\le d}$.
	Then, as mentioned in the proof of Theorem 5.1 in \cite{Gonzalez-Pinto}, $\mu(F^{\rm 2nd})<0$.
	Besides, since $F^{\rm 1st}$ is anti-symmetric, $\mu(F^{\rm 1st})=0$.
	Combining this,
	\begin{equation}
		\mu(F)\le \mu(F^{\rm 1st}) + \mu(F^{\rm 2nd})<0.
	\end{equation}
\end{proof}

Using this, we can prove Lemma \ref{lem:FDMErr}.

\begin{proof}[Proof of Lemma \ref{lem:FDMErr}]
	Because of (\ref{eq:UDerivCond}), for $i\in[d]$,
	\begin{equation}
		\left|\frac{\partial}{\partial x_i}Y(\tau,\vec{x})-\frac{Y(\tau,\vec{x}+h_i\vec{e}_i)-Y(\tau,\vec{x}-h_i\vec{e}_i)}{2h_i}\right|<\frac{\zeta}{6}h_i^2,
	\end{equation}
	and
	\begin{equation}
		\left|\frac{\partial^2}{\partial x_i^2}Y(\tau,\vec{x})-\frac{Y(\tau,\vec{x}+h_i\vec{e}_i)-2Y(\tau,\vec{x})+Y(\tau,\vec{x}-h_i\vec{e}_i)}{h_i^2}\right|<\frac{\xi}{12}h_i^2.
	\end{equation}
	hold, and, for $i,j\in [d]$ such that $i\ne j$,
	\begin{equation}
		\left|\frac{\partial^2}{\partial x_i\partial x_j}Y(\tau,\vec{x})-\frac{Y(\tau,\vec{x}+h_i\vec{e}_i+h_j\vec{e}_j)-Y(\tau,\vec{x}+h_i\vec{e}_i-h_j\vec{e}_j)-Y(\tau,\vec{x}-h_i\vec{e}_i+h_j\vec{e}_j)+Y(\tau,\vec{x}-h_i\vec{e}_i-h_j\vec{e}_j)}{4h_ih_j}\right|<\frac{\xi}{6}h_ih_j
	\end{equation}
	holds, where $\vec{e}_i,i\in[d]$ is the $d$-dimensional vector whose $i$-th element is 1 and the others are 0. 
	Therefore, we see that
	\begin{eqnarray}
		|\mathcal{L}Y(\tau,\vec{x}^{(k)})-(F\vec{Y}(\tau)+\vec{C}(\tau))_k|&<&\sum_{i=1}^d \frac{\xi}{24}\sigma_i^2h_i^2 + \sum_{i=1}^d\sum_{j=i+1}^d \frac{\xi}{6}\sigma_i\sigma_j|\rho_{ij}|h_ih_j + \sum_{i=1}^d \frac{\zeta}{6}\left|r-\frac{1}{2}\sigma_i^2\right|h_i^2 \nonumber \\
		&<&\frac{\epsilon}{T}, \label{eq:diffusionDiff}
	\end{eqnarray}	
	where $(F\vec{Y}(\tau)+\vec{C}(\tau))_k$ is the $k$-th element of $F\vec{Y}(\tau)+\vec{C}(\tau)$ and we used (\ref{eq:hxcond2}).
	Since
	\begin{equation}
		\frac{d}{d\tau}(\vec{\Tilde{Y}}(\tau)-\vec{Y}(\tau))=F\vec{\Tilde{Y}}(\tau)+\vec{C}(\tau)-\mathcal{L}\vec{Y}(\tau)=F(\vec{\Tilde{Y}}(\tau)-\vec{Y}(\tau))+F\vec{Y}(\tau)+\vec{C}(\tau)-\mathcal{L}\vec{Y}(\tau),
	\end{equation}
	where $\mathcal{L}\vec{Y}(\tau):=(\mathcal{L}Y(\tau,\vec{x}^{(1)}),...,\mathcal{L}Y(\tau,\vec{x}^{(N_{\rm gr})}))^T$, we finally obtain
	\begin{equation}
		\|\vec{\Tilde{Y}}(\tau)-\vec{Y}(\tau)\|\le  e^{\tau\mu(F)}\|\vec{\Tilde{Y}}(0)-\vec{Y}(0)\|+\int_0^{\tau} e^{(\tau-\tau^\prime)\mu(F)}\|F\vec{Y}(\tau^\prime)+\vec{C}(\tau^\prime)-\mathcal{L}\vec{Y}(\tau^\prime)\|d\tau^\prime\le \sqrt{N_{\rm gr}}\epsilon
	\end{equation}
	by (2.1) in \cite{Soderlind}.
	Here, we used $\vec{\Tilde{Y}}(0)=\vec{Y}(0),\mu(F)<0$ and (\ref{eq:diffusionDiff}).

\end{proof}

\subsection{Proof of Lemma \ref{lem:tter} \label{sec:PrLemtter}}

\subsubsection{Upper bound the probability that the underlying asset prices reach the boundaries}

In order to prove Lemma \ref{lem:tter}, we weed some subsidiary lemmas.
First, we prove the following one on the probability that the underlying asset prices reach the boundaries.

\begin{lemma}
	Let $\epsilon$ be a positive real number.
	For $S_i,i\in[d]$ in (\ref{eq:SDE}) and any $t\in(0,\Tilde{t}_{\rm u}]$, where
	\begin{equation}
		\Tilde{t}_{\rm u} := \frac{\left(\log \left(\frac{H_i}{S_{i,0}}\right)\right)^2}{\sigma^2\log\epsilon^{-1}}, \label{eq:ttil}
	\end{equation}
	the following holds
	\begin{equation}
			P \left(\max_{0\le s \le t}S_i(s) \ge H_i \ \middle| \ S_i(t)=s\right) \le \epsilon  \ {\rm if} \ s < \sqrt{S_{i,0}H_i}. \label{eq:PH}
	\end{equation}
	Similarly, for any $t\in(0,\Tilde{t}_{\rm l})$, where
	\begin{equation}
		\Tilde{t}_{\rm l} := \frac{\left(\log \left(\frac{S_{i,0}}{L_i}\right)\right)^2}{\sigma^2\log\epsilon^{-1}},
	\end{equation}
	the following holds
	\begin{equation}
			P \left(\min_{0\le s \le t}S_i(s) \le L_i \ \middle| \ S_i(t)=s\right) \le \epsilon  \ {\rm if} \ s > \sqrt{S_{i,0}L_i}. \label{eq:PL}
	\end{equation}
	\label{lemma:BBMax}
\end{lemma}

\begin{proof}
	It is well-known (see e.g. \cite{Shreve}) that $S_i(t)$ can be written as
	\begin{equation}
		S_i(t) = S_{i,0} \exp\left(\sigma_i W_i(t) -\left(\frac{1}{2}\sigma_i^2-r\right)t\right). \label{eq:Si}
	\end{equation}
	Therefore, we see that
	\begin{equation}
		S_i(t)=H_i \Leftrightarrow B_i(t):=W_i(t) - \left(\frac{\sigma_i}{2}-\frac{r}{\sigma_i}\right)t = \frac{\log \left(\frac{H_i}{S_{i,0}}\right)}{\sigma_i}
	\end{equation}
	and
	\begin{equation}
		S_i(t)=\sqrt{S_{i,0}H_i} \Leftrightarrow B_i(t) = \frac{\log \left(\frac{H_i}{S_{i,0}}\right)}{2\sigma_i}.
	\end{equation}
	Using a formula on the distribution of the maximum of a Brownian bridge with drift (THEOREM 3.1 in \cite{Beghin}), we obtain
	\begin{equation}
		P \left(\max_{0\le s \le t}B_i(s) \ge \frac{\log \left(\frac{H_i}{S_{i,0}}\right)}{\sigma_i} \middle| B_i(t)=\frac{\log \left(\frac{H_i}{S_{i,0}}\right)}{2\sigma_i}\right)= \exp\left(-\frac{2}{t}\frac{\log \left(\frac{H_i}{S_{i,0}}\right)}{\sigma_i}\frac{\log \left(\frac{H_i}{S_{i,0}}\right)}{2\sigma_i}\right).
	\end{equation}
	Then, for $t\in(0,\Tilde{t}_{\rm u}]$, we obtain (\ref{eq:PH}).
	The later part of the statement is proven similarly.
\end{proof}

Using Lemma \ref{lemma:BBMax}, we can prove the following.

\begin{lemma}
	Consider $S_1,...,S_d$ in (\ref{eq:SDE}).
	Let $\epsilon$ be a positive real number.
	Then, for any $\vec{S}:=(s_1,...,s_d)^T\in D_{\rm half}$, where
	\begin{equation}
		D_{\rm half}:=(\sqrt{L_1S_{1,0}},\sqrt{U_1S_{1,0}})\times\cdots\times(\sqrt{L_dS_{d,0}},\sqrt{U_dS_{d,0}}) \label{eq:Dhalf},
	\end{equation}
	 and any $t\in (0,t_b]$, where
	\begin{equation}
		t_b := \min \left\{\frac{\left(\log \left(\frac{U_1}{S_{1,0}}\right)\right)^2}{\sigma_1^2\log\left(\frac{2d}{\epsilon}\right)},...,\frac{\left(\log \left(\frac{U_d}{S_{d,0}}\right)\right)^2}{\sigma_d^2\log\left(\frac{2d}{\epsilon}\right)},\frac{\left(\log \left(\frac{S_{1,0}}{L_1}\right)\right)^2}{\sigma_1^2\log\left(\frac{2d}{\epsilon}\right)},...,\frac{\left(\log \left(\frac{S_{d,0}}{L_d}\right)\right)^2}{\sigma_d^2\log\left(\frac{2d}{\epsilon}\right)}\right\},
	\end{equation}
	the following holds
	\begin{equation}
		p_{\rm NB}(t,\vec{S}) \ge 1-\epsilon,
	\end{equation}
	where $p_{\rm NB}(t,\vec{S})$ is defined below (\ref{eq:V0ExpPrice}).
	\label{lemma:phibound}
\end{lemma}

\begin{proof}
	\begin{eqnarray}
		&& p_{\rm NB}(t,\vec{S}) \nonumber \\
		&\ge& P\left(\vec{S}(u) {\rm \ does \ not \ reach \ any \ boundaries \ by \ } t \ \middle| \ \vec{S}_{t}=\vec{S}\right) \nonumber \\
		&=& 1-P\left(\vec{S}(u) {\rm \ reaches \ either \ of \ boundaries \ by \ } t \ \middle| \ \vec{S}_{t}=\vec{S}\right) \nonumber \\
		&\ge& 1-\sum_{i=1}^dP\left(\max_{0\le u \le t}S_{i}(u) \ge H_i  \ \middle| \ S_i(t)=s_i\right)  - \sum_{i=1}^dP\left(\min_{0\le u \le t}S_{i}(u) \le L_i  \ \middle| \ S_i(t)=s_i\right) \nonumber \\
		&\ge& 1-d\times\frac{\epsilon}{2d}-d\times\frac{\epsilon}{2d}=1-\epsilon,
	\end{eqnarray}
	where we used Lemma \ref{lemma:BBMax} at the last inequality.
\end{proof}

\subsubsection{Upper bound the integral on the outside of the boundaries}

Besides, we need the following lemmas, in order to upper bound the contribution from the outside of the boundaries to the integral (\ref{eq:V0ExpPrice}).

\begin{lemma}
	Consider $S_i,i\in[d]$ in (\ref{eq:SDE}).
	Let $H$ be a real number such that $H>S_{i,0}$ and $\epsilon$ be a positive real number satisfying
	\begin{equation}
		\log\left(\frac{1}{2\epsilon}\right)>\frac{4}{5}\left(1+\frac{2r}{\sigma_i^2}\right)\log\left(\frac{H}{S_{i,0}}\right). \label{eq:cutcond}
	\end{equation}
	Then, for any $t\in(0,t_{\rm cu})$, 
	\begin{equation}
		\int^\infty_{H} s \phi_i(t,s)ds <\epsilon S_{i,0}e^{rt},  \int^\infty_{H} \phi_i(t,s)ds <\epsilon
	\end{equation}
	holds, where $\phi_i(t,s)$ is the probability density of $S_i(t)$ and
	\begin{equation}
		t_{\rm cu} := \frac{8\left(\log \left(\frac{H}{S_{i,0}}\right)\right)^2}{25\sigma_i^2 \log \left(\frac{1}{2\epsilon}\right)}. \label{eq:tcu}
	\end{equation}
	\label{lemma:cutu}
\end{lemma}

\begin{proof}
	Because of (\ref{eq:Si}) and the basic property of the Brownian motion, the probability density of $x_i(t)=\log S_i(t)$ is
	\begin{equation}
		\frac{1}{\sqrt{2\pi t}\sigma_i}\exp\left(-\frac{1}{2\sigma_i^2t}\left(x-\left(r-\frac{1}{2}\sigma_i^2\right)t\right)^2\right).
	\end{equation}
	Therefore, we see that
	\begin{eqnarray}
		&&\int^\infty_{H} s \phi_i(t,s)ds \nonumber \\
		&=& \int^\infty_{\log(H/S_{i,0})} e^x \frac{1}{\sqrt{2\pi t}\sigma_i}\exp\left(-\frac{1}{2\sigma_i^2t}\left(x-\left(r-\frac{1}{2}\sigma_i^2\right)t\right)^2\right)dx \nonumber \\
		&=& \frac{S_{i,0}}{\sqrt{2\pi t}\sigma_i}e^{rt}\int^\infty_{\log(H/S_{i,0})} \exp\left(-\frac{1}{2\sigma_i^2t}\left(x-\left(r+\frac{1}{2}\sigma_i^2\right)t\right)^2\right)dx \nonumber \\
		&<& \frac{S_{i,0}e^{rt}}{2}\exp\left(-\frac{1}{2\sigma_i^2t}\left(\log\left(\frac{H}{S_{i,0}}\right)-\left(r+\frac{\sigma_i^2}{2}\right)t\right)^2\right). \label{eq:tocyuu}
	\end{eqnarray}
	Here, we used
	\begin{equation}
	\frac{2}{\sqrt{\pi}} \int^{\infty}_{c} e^{-y^2} dy <e^{-c^2}, \label{eq:GauIntUB}
	\end{equation}
	which hold for any $c\in\mathbb{R}_+$.
	Besides, because of (\ref{eq:cutcond}) and (\ref{eq:tcu}),
	\begin{equation}
		\left(r+\frac{\sigma^2}{2}\right)t < \left(r+\frac{\sigma^2}{2}\right)\frac{8\left(\log \left(\frac{H}{S_{i,0}}\right)\right)^2}{25\sigma^2 \log \left(\frac{1}{2\epsilon}\right)}<\frac{1}{5}\log \left(\frac{H}{S_{i,0}}\right)
		\label{eq:tocyuu2}
	\end{equation}
	holds for $t\in(0,t_{\rm cu})$.
	Combining (\ref{eq:tcu}), (\ref{eq:tocyuu}) and (\ref{eq:tocyuu2}), we obtain
	\begin{equation}
		\int^\infty_{H} s \phi_i(t_{\rm cu},s)ds < \frac{S_{i,0}e^{rt}}{2}\exp\left(-\frac{1}{2\sigma_i^2t}\frac{16}{25}\left(\log \left(\frac{H}{S_{i,0}}\right)\right)^2\right)<\epsilon S_{i,0}e^{rt}
	\end{equation}
	for $t\in(0,t_{\rm cu})$.
	
	On the other hand,
	\begin{eqnarray}
		&&\int^\infty_{H} \phi_i(t_{\rm cu},s)ds \nonumber \\
		&=& \int^\infty_{\log(H/S_0)} \frac{1}{\sqrt{2\pi t}\sigma}\exp\left(-\frac{1}{2\sigma^2t}\left(x-\left(r-\frac{1}{2}\sigma^2\right)t\right)^2\right)dx \nonumber \\
		&<&\frac{1}{2}\exp\left(-\frac{1}{2\sigma^2t}\left(\log\left(\frac{H}{S_0}\right)-\left(r-\frac{\sigma^2}{2}\right)t\right)^2\right),
	\end{eqnarray}
	where we used (\ref{eq:GauIntUB}) again.
	Combining this and $\left(r-\frac{\sigma_i^2}{2}\right)t <\frac{1}{5}\log \left(\frac{H}{S_{i,0}}\right)$, which holds for $t\in(0,t_{\rm cu})$ because of (\ref{eq:tocyuu2}), we obtain
	\begin{equation}
		\int^\infty_{H} \phi(t_{\rm cu},s)ds < \frac{1}{2}\exp\left(-\frac{1}{2\sigma_i^2t}\frac{16}{25}\left(\log \left(\frac{H}{S_{i,0}}\right)\right)^2\right)<\epsilon.
	\end{equation}
	
\end{proof}

\begin{lemma}
	Consider $S_i,i\in[d]$ in (\ref{eq:SDE}).
	Let $L$ be a real number such that $L<S_{i,0}$ and $\epsilon$ be a positive real number satisfying
	\begin{equation}
		\log\left(\frac{1}{2\epsilon}\right)>\frac{4}{5}\left(1-\frac{2r}{\sigma_i^2}\right)\log\left(\frac{S_{i,0}}{L}\right). \label{eq:cutcond2}
	\end{equation}
	Then, for any $t\in(0,t_{\rm cl})$, 
	\begin{equation}
	\int^{L}_0 s \phi_i(t,s)ds <\epsilon S_{i,0}e^{rt},  \int^{L}_0 \phi_i(t,s)ds <\epsilon   
	\end{equation}
	holds, where $\phi_i(t,s)$ is the probability density of $S_i(t)$ and
	\begin{equation}
		t_{\rm cl} := \frac{8\left(\log \left(\frac{S_{i,0}}{L}\right)\right)^2}{25\sigma_i^2 \log \left(\frac{1}{2\epsilon}\right)}.  \label{eq:tcl}
	\end{equation}
	\label{lemma:cutl}
\end{lemma}

\begin{proof}
	Similarly to (\ref{eq:tocyuu}), for $t\in(0,t_{\rm cl})$,
	\begin{eqnarray}
		&&\int^{L}_0 s \phi_i(t,s)ds \nonumber \\
		&=& \int^{\log(L/S_{i,0})}_{-\infty} e^x \frac{1}{\sqrt{2\pi t}\sigma_i}\exp\left(-\frac{1}{2\sigma_i^2t}\left(x-\left(r-\frac{1}{2}\sigma_i^2\right)t\right)^2\right)dx \nonumber \\
		&=& \frac{S_{i,0}}{\sqrt{2\pi t}\sigma_i}e^{rt}\int^{\log(L/S_{i,0})}_{-\infty} \exp\left(-\frac{1}{2\sigma_i^2t}\left(x-\left(r+\frac{1}{2}\sigma_i^2\right)t\right)^2\right)dx \nonumber \\
		&<& \frac{S_{i,0}e^{rt}}{2}\exp\left(-\frac{1}{2\sigma_{i}^2t}\left(\log\left(\frac{S_{i,0}}{L}\right)+\left(r+\frac{\sigma_i^2}{2}\right)t\right)^2\right) \nonumber \\
		&<&  \frac{S_{i,0}e^{rt}}{2}\exp\left(-\frac{1}{2\sigma_i^2t}\left(\log\left(\frac{S_{i,0}}{L}\right)\right)^2\right) \nonumber \\ 
		&<& \epsilon S_{i,0} e^{rt},
	\end{eqnarray}
	where we used (\ref{eq:GauIntUB}) at the first inequality and (\ref{eq:tcl}) at the last inequality.
	
	On the other hand,
	\begin{eqnarray}
		&&\int^L_0 \phi(t,s)ds \nonumber \\
		&=& \int^{\log(L/S_{i,0})}_{-\infty} \frac{1}{\sqrt{2\pi t}\sigma_i}\exp\left(-\frac{1}{2\sigma_i^2t}\left(x-\left(r-\frac{1}{2}\sigma_i^2\right)t\right)^2\right)dx \nonumber \\
		&<&\frac{1}{2}\exp\left(-\frac{1}{2\sigma_i^2t}\left(\log\left(\frac{S_{i,0}}{L}\right)+\left(r-\frac{\sigma_i^2}{2}\right)t\right)^2\right), \label{eq:tocyuu3}
	\end{eqnarray}
	where we used (\ref{eq:GauIntUB}) again.
	Then, for $t\in(0,t_{\rm cl})$, (\ref{eq:tocyuu3}) and 
	\begin{equation}
		\left(\frac{\sigma_i^2}{2}-r\right)t < \left(\frac{\sigma_i^2}{2}-r\right)\frac{8\left(\log \left(\frac{S_{i,0}}{L}\right)\right)^2}{25\sigma_i^2 \log \left(\frac{1}{2\epsilon}\right)}<\frac{1}{5}\log \left(\frac{S_{i,0}}{L}\right),
	\end{equation}
	which follows (\ref{eq:cutcond2}), lead to
	\begin{equation}
		\int^{L}_0 \phi_i(t,s)ds< \frac{1}{2}\exp\left(-\frac{1}{2\sigma_i^2t}\frac{16}{25}\left(\log \left(\frac{S_{i,0}}{L}\right)\right)^2\right)<\epsilon.
	\end{equation}
	
\end{proof}

Combining these lemma, we obtain the following.

\begin{lemma}
	Consider $S_1,...,S_d$ in (\ref{eq:SDE}) under Assumption \ref{ass:payoffBound}.
	For any $\epsilon\in\mathbb{R}_+$ satisfying
	\begin{equation}
		\log\left(\frac{\Tilde{A}d(d+1)}{\epsilon}\right)>\max\left\{\frac{2}{5}\left(1-\frac{2r}{\sigma_i^2}\right)\log\left(\frac{U_i}{S_{i,0}}\right),\frac{2}{5}\left(1-\frac{2r}{\sigma_i^2}\right)\log\left(\frac{S_{i,0}}{L_i}\right)\right\},i=1,...,d,
	\end{equation}
	where $\Tilde{A}=\max\{A_1\sqrt{U_1S_{1,0}},...,A_d\sqrt{U_dS_{d,0}},A_0\}$, the following holds
	\begin{equation}
		e^{-rt_{\rm ter}}\int_{\mathbb{R}_+^d \setminus D_{\rm half}} d\vec{S} \phi(t_{\rm ter},\vec{S})V(t_{\rm ter},\vec{S})\le\epsilon, \label{eq:VIntegTail}
	\end{equation}
	where $t_{\rm ter}$ is defined as (\ref{eq:tter}) and $D_{\rm half}$ is defined as (\ref{eq:Dhalf}).
	\label{lemma:VIntegBound}
\end{lemma}

\begin{proof}
	First, note that, under Assumption \ref{ass:payoffBound}, for $\vec{S}=(S_1,...,S_d)^T\in \mathbb{R}_+^d$,
	\begin{eqnarray}
		V(t_{\rm ter},\vec{S})&=&E[e^{-r(T-t_{\rm ter})}f_{\rm pay}(\vec{S}(T))1_{\rm NB}|\vec{S}(t_{\rm ter})=\vec{S}] \nonumber \\
		&\le& E\left[e^{-r(T-t_{\rm ter})}\left(\sum_{i=1}^d A_iS_i(T)+A_0\right)\middle|\vec{S}(t_{\rm ter})=\vec{S}\right] \nonumber \\
		&=& \sum_{i=1}^d A_iS_i+A_0e^{-r(T-t_{\rm ter})}. \label{eq:VBound}
	\end{eqnarray}
	Therefore, we obtain
	\begin{equation}
		e^{-rt_{\rm ter}}\int_{\mathbb{R}_+^d \setminus D_{\rm half}} d\vec{S} \phi(t_{\rm ter},\vec{S})V(t_{\rm ter},\vec{S}) \le \sum_{i=1}^d A_ie^{-rt_{\rm ter}}\int_{\mathbb{R}_+^d \setminus D_{\rm half}} d\vec{S} S_i\phi(t_{\rm ter},\vec{S}) + A_0e^{-rT}\int_{\mathbb{R}_+^d \setminus D_{\rm half}} d\vec{S} \phi(t_{\rm ter},\vec{S}).
	\end{equation}
	We can evaluate $A_1e^{-rt_{\rm ter}}\int_{\mathbb{R}_+^d \setminus D_{\rm half}} d\vec{S} S_1\phi(t_{\rm ter},\vec{S})$ as follows
	\begin{eqnarray}
		A_1e^{-rt_{\rm ter}}\int_{\mathbb{R}_+^d \setminus D_{\rm half}} d\vec{S} S_1\phi(t_{\rm ter},\vec{S}) &\le& A_1e^{-rt_{\rm ter}}\int_{S_1\ge \sqrt{U_1S_{1,0}}} d\vec{S} S_1\phi(t_{\rm ter},\vec{S}) \nonumber \\
		&+& A_1e^{-rt_{\rm ter}}\int_{S_1\le \sqrt{L_1S_{1,0}}} d\vec{S} S_1\phi(t_{\rm ter},\vec{S}) \nonumber \\
		&+&\sum_{i=2}^d A_1e^{-rt_{\rm ter}}\int_{\substack{\sqrt{L_1S_{1,0}}\le S_1\le \sqrt{U_1S_{1,0}} \\ S_i \ge \sqrt{U_iS_{i,0}} }} d\vec{S} S_1\phi(t_{\rm ter},\vec{S})\nonumber \\
		&+&\sum_{i=2}^d A_1e^{-rt_{\rm ter}}\int_{\substack{\sqrt{L_1S_{1,0}}\le S_1\le \sqrt{U_1S_{1,0}} \\ S_i \le \sqrt{L_iS_{i,0}}}} d\vec{S} S_1\phi(t_{\rm ter},\vec{S}).
	\end{eqnarray}
	In the right hand side, the first term is $A_1e^{-rt_{\rm ter}}\int^\infty_{\sqrt{U_1S_{1,0}}} dS_1 S_1\phi_1(t_{\rm ter},S_1)$, where $\phi_i(t,S_i)$ is the marginal density of $S_i(t)$, and therefore
	\begin{equation}
		A_1e^{-rt_{\rm ter}}\int_{S_1\ge \sqrt{U_1S_{1,0}}} d\vec{S} S_1\phi(t_{\rm ter},\vec{S}) \le \frac{\epsilon S_{1,0}A_1}{2d(d+1) \Tilde{A}}\le \frac{\epsilon}{2d(d+1)}
	\end{equation}
	holds from Lemma \ref{lemma:cutu}\footnote{Note that \begin{equation}\frac{8\left(\log \left(\frac{\sqrt{U_iS_{i,0}}}{S_{i,0}}\right)\right)^2}{25\sigma_i^2 \log \left(\frac{2\Tilde{A}d(d+1)}{\epsilon}\right)}=\frac{2\left(\log \left(\frac{U_i}{S_{i,0}}\right)\right)^2}{25\sigma_i^2 \log \left(\frac{2\Tilde{A}d(d+1)}{\epsilon}\right)},\frac{8\left(\log \left(\frac{S_{i,0}}{\sqrt{L_iS_{i,0}}}\right)\right)^2}{25\sigma_i^2 \log \left(\frac{2\Tilde{A}d(d+1)}{\epsilon}\right)}=\frac{2\left(\log \left(\frac{S_{i,0}}{L_i}\right)\right)^2}{25\sigma_i^2 \log \left(\frac{2\Tilde{A}d(d+1)}{\epsilon}\right)}.\nonumber\end{equation}}.
	Similarly, from Lemma \ref{lemma:cutl}, the second term is bounded as
	\begin{equation}
		A_1e^{-rt_{\rm ter}}\int_{S_1\le \sqrt{L_1S_{1,0}}} d\vec{S} S_1\phi(t_{\rm ter},\vec{S}) \le \frac{\epsilon }{2d(d+1)}.
	\end{equation}
	On the other hand, from Lemma \ref{lemma:cutu}, we see that the third term is bounded as
	\begin{eqnarray}
		&& \sum_{i=2}^d A_1e^{-rt_{\rm ter}}\int_{\substack{\sqrt{L_1S_{1,0}}\le S_1\le \sqrt{U_1S_{1,0}} \\ S_i \ge \sqrt{U_iS_{i,0}}}} d\vec{S} S_1\phi(t_{\rm ter},\vec{S}) \nonumber \\
		&\le& \sum_{i=2}^d A_1\sqrt{U_1S_{1,0}}\int_{\substack{\sqrt{L_1S_{1,0}}\le S_1\le \sqrt{U_1S_{1,0}} \\ S_i \ge \sqrt{U_iS_{i,0}}}} d\vec{S} \phi(t_{\rm ter},\vec{S})  \nonumber \\
		&\le& \sum_{i=2}^d A_1\sqrt{U_1S_{1,0}}\int^\infty_{\sqrt{U_iS_{i,0}}} ds \phi_i(t_{\rm ter},s)  \nonumber \\
		&\le& \sum_{i=2}^d \frac{\epsilon A_1\sqrt{U_1S_{1,0}}}{2d(d+1) \Tilde{A}}  \nonumber \\
		&\le& \sum_{i=2}^d \frac{\epsilon}{2d(d+1)}  \nonumber \\
		&=& \frac{\epsilon(d-1)}{2d(d+1)}
	\end{eqnarray}
	and, similarly, the fourth term is bounded as
	\begin{equation}
		\sum_{i=2}^d A_1e^{-rt_{\rm ter}}\int_{\substack{\sqrt{L_1S_{1,0}}\le S_1\le \sqrt{U_1S_{1,0}} \\ S_i \le \sqrt{L_iS_{i,0}}}} d\vec{S} S_1\phi_{\vec{S}}(t_{\rm ter},\vec{S}) \le \frac{\epsilon(d-1)}{2d(d+1)}
	\end{equation}
	by Lemma \ref{lemma:cutl}.
	In summary,
	\begin{equation}
		A_1e^{-rt_{\rm ter}}\int_{\mathbb{R}_+^d \setminus D_{\rm half}} d\vec{S} S_1\phi(t_{\rm ter},\vec{S}) \le \frac{\epsilon}{d+1}
	\end{equation}
	holds.
	$A_2e^{-rt_{\rm ter}}\int_{\mathbb{R}_+^d \setminus D_{\rm half}} d\vec{S} S_2\phi(t_{\rm ter},\vec{S}),...,A_de^{-rt_{\rm ter}}\int_{\mathbb{R}_+^d \setminus D_{\rm half}} d\vec{S} S_d\phi(t_{\rm ter},\vec{S})$ are bounded similarly.
	
	On the other hand, by Lemmas \ref{lemma:cutu} and \ref{lemma:cutl},
	\begin{eqnarray}
		&&A_0e^{-rt_{\rm ter}}\int_{\mathbb{R}_+^d \setminus D_{\rm half}} d\vec{S} \phi(t_{\rm ter},\vec{S}) \nonumber \\
		&\le& \sum_{i=1}^d \left(A_0e^{-rt_{\rm ter}}\int_{S_i \ge \sqrt{U_iS_{i,0}}} d\vec{S} \phi(t_{\rm ter},\vec{S})+A_0e^{-rt_{\rm ter}}\int_{S_i \le \sqrt{L_iS_{i,0}}} d\vec{S} \phi(t_{\rm ter},\vec{S})\right) \nonumber \\
		&\le& \sum_{i=1}^d \frac{\epsilon A_0e^{-rt_{\rm ter}}}{d(d+1) \Tilde{A}} \nonumber \\
		&\le& \frac{\epsilon}{d+1}
	\end{eqnarray}
	holds.
	
	Summing up all terms, we obtain (\ref{eq:VIntegTail}).
\end{proof}

\subsubsection{Proof of Lemma \ref{lem:tter}}

Then, we finally prove Lemma \ref{lem:tter}.

\begin{proof}[The proof of Lemma \ref{lem:tter}]
	Note that $t_{\rm ter}$ satisfies
	\begin{eqnarray}
		t_{\rm ter} &<& \min \left\{\frac{\left(\log \left(\frac{U_1}{S_{1,0}}\right)\right)^2}{\sigma_1^2\log\left(\frac{2d(d+1)\Tilde{\Tilde{A}}}{\epsilon}\right)},...,\frac{\left(\log \left(\frac{U_d}{S_{d,0}}\right)\right)^2}{\sigma_d^2\log\left(\frac{2d(d+1)\Tilde{\Tilde{A}}}{\epsilon}\right)},\frac{\left(\log \left(\frac{S_{1,0}}{L_1}\right)\right)^2}{\sigma_1^2\log\left(\frac{2d(d+1)\Tilde{\Tilde{A}}}{\epsilon}\right)},...,\frac{\left(\log \left(\frac{S_{d,0}}{L_d}\right)\right)^2}{\sigma_d^2\log\left(\frac{2d(d+1)\Tilde{\Tilde{A}}}{\epsilon}\right)}\right\}, \label{eq:tterLtThres}
	\end{eqnarray}
	where $\Tilde{\Tilde{A}} := \max\{A_0,A_1S_{1,0},...,A_dS_{d,0}\}$.
	Besides, we can see that
	\begin{eqnarray}
		&& \left|V(0,\vec{S}_0) - e^{-rT}\int_{\Tilde{\Tilde{D}}} d\vec{x} \Tilde{\phi}(t_{\rm ter},\vec{x})Y(\tau_{\rm ter},\vec{x})\right| \nonumber \\
		&=& \left|V(0,\vec{S}_0) - e^{-rt_{\rm ter}}\int_{\hat{D}} d\vec{S} \phi(t_{\rm ter},\vec{S})V(t_{\rm ter},\vec{S})\right| \nonumber \\
		&=& \left|e^{-rt_{\rm ter}}\int d\vec{S}_{\mathbb{R}_+^d} \phi(t_{\rm ter},\vec{S})p_{\rm NB}(t_{\rm ter},\vec{S})V(t_{\rm ter},\vec{S})- e^{-rt_{\rm ter}}\int_{\hat{D}} d\vec{S} \phi(t_{\rm ter},\vec{S})V(t_{\rm ter},\vec{S})\right| \nonumber \\
		&=& \left|e^{-rt_{\rm ter}}\int_{D_{\rm half}} d\vec{S} \phi(t_{\rm ter},\vec{S})p_{\rm NB}(t_{\rm ter},\vec{S})V(t_{\rm ter},\vec{S}) + e^{-rt_{\rm ter}}\int_{\mathbb{R}_+^d\setminus D_{\rm half}} d\vec{S} \phi(t_{\rm ter},\vec{S})p_{\rm NB}(t_{\rm ter},\vec{S})V(t_{\rm ter},\vec{S}) \right. \nonumber \\
		&& \qquad - \left.e^{-rt_{\rm ter}}\int_{D_{\rm half}} d\vec{S} \phi(t_{\rm ter},\vec{S})V(t_{\rm ter},\vec{S})- e^{-rt_{\rm ter}}\int_{\hat{D}\setminus D_{\rm half}} d\vec{S} \phi(t_{\rm ter},\vec{S})V(t_{\rm ter},\vec{S}) \right| \nonumber \\
		&\le& \left|e^{-rt_{\rm ter}}\int_{D_{\rm half}} d\vec{S} \phi(t_{\rm ter},\vec{S})p_{\rm NB}(t_{\rm ter},\vec{S})V(t_{\rm ter},\vec{S})- e^{-rt_{\rm ter}}\int_{D_{\rm half}} d\vec{S} \phi(t_{\rm ter},\vec{S})V(t_{\rm ter},\vec{S})\right| \nonumber \\
		&& \qquad + \left|e^{-rt_{\rm ter}}\int_{\mathbb{R}_+^d\setminus D_{\rm half}} d\vec{S} \phi(t_{\rm ter},\vec{S})p_{\rm NB}(t_{\rm ter},\vec{S})V(t_{\rm ter},\vec{S})-e^{-rt_{\rm ter}}\int_{\hat{D}\setminus D_{\rm half}} d\vec{S} \phi(t_{\rm ter},\vec{S})V(t_{\rm ter},\vec{S})\right|,\nonumber \\
		&& \label{eq:tocyuu4}
	\end{eqnarray}
	where
	\begin{equation}
		\hat{D} := \left[\exp\left(\frac{1}{2}\left(l_1+x^{(0)}_1\right)\right),\exp\left(\frac{1}{2}\left(x^{(n_{\rm gr}-1)}_1+u_1\right)\right)\right]\times \cdots \times \left[\exp\left(\frac{1}{2}\left(l_d+x^{(0)}_d\right)\right),\exp\left(\frac{1}{2}\left(x^{(n_{\rm gr}-1)}_d+u_d\right)\right)\right].
	\end{equation}
	The first term in the last line in (\ref{eq:tocyuu4}) is bounded as
	\begin{eqnarray}
		&&\left|e^{-rt_{\rm ter}}\int_{D_{\rm half}} d\vec{S} \phi(t_{\rm ter},\vec{S})p_{\rm NB}(t_{\rm ter},\vec{S})V(t_{\rm ter},\vec{S})- e^{-rt_{\rm ter}}\int_{D_{\rm half}} d\vec{S} \phi(t_{\rm ter},\vec{S})V(t_{\rm ter},\vec{S})\right| \nonumber \\
		&=& e^{-rt_{\rm ter}}\int_{D_{\rm half}} d\vec{S} \phi(t_{\rm ter},\vec{S})(1-p_{\rm NB}(t_{\rm ter},\vec{S}))V(t_{\rm ter},\vec{S}) \nonumber \\
		&\le& \frac{e^{-rt_{\rm ter}}\epsilon}{\Tilde{\Tilde{A}}(d+1)}\int_{D_{\rm half}} d\vec{S}\phi_{\vec{S}}(t_{\rm ter},\vec{S})V(t_{\rm ter},\vec{S}) \nonumber \\
		&\le& \frac{e^{-rt_{\rm ter}}\epsilon}{\Tilde{\Tilde{A}}(d+1)} \int_{\mathbb{R}_+^d} d\vec{S}\phi_{\vec{S}}(t_{\rm ter},\vec{S})V(t_{\rm ter},\vec{S}) \nonumber \\
		&\le& \frac{e^{-rt_{\rm ter}}\epsilon}{\Tilde{\Tilde{A}}(d+1)}\int_{\mathbb{R}_+^d} d\vec{S}\phi_{\vec{S}}(t_{\rm ter},\vec{S})\left(\sum_{i=1}^dA_iS_i + A_0e^{-r(T-t_{\rm ter})}\right) \nonumber \\
		&=&\frac{\epsilon}{\Tilde{\Tilde{A}}(d+1)}\left(\sum_{i=1}^dA_iS_{i,0} + A_0e^{-r(T-t_{\rm ter})}\right) \nonumber \\
		&\le& \epsilon,
	\end{eqnarray}
	where we used Lemma \ref{lemma:phibound} and (\ref{eq:tterLtThres}) at the first inequality, and (\ref{eq:VBound}) at the third inequality.
	On the other hand, the second term of (\ref{eq:tocyuu4}) is bounded as
	\begin{eqnarray}
		&&\left|e^{-rt_{\rm ter}}\int_{\mathbb{R}_+^d\setminus D_{\rm half}} d\vec{S} \phi(t_{\rm ter},\vec{S})p_{\rm NB}(t_{\rm ter},\vec{S})V(t_{\rm ter},\vec{S})-e^{-rt_{\rm ter}}\int_{\hat{D}\setminus D_{\rm half}} d\vec{S} \phi(t_{\rm ter},\vec{S})V(t_{\rm ter},\vec{S})\right| \nonumber \\
		&=& \left|e^{-rt_{\rm ter}}\int_{\mathbb{R}_+^d\setminus \hat{D}} d\vec{S} \phi(t_{\rm ter},\vec{S})p_{\rm NB}(t_{\rm ter},\vec{S})V(t_{\rm ter},\vec{S})-e^{-rt_{\rm ter}}\int_{\hat{D}\setminus D_{\rm half}} d\vec{S} (1-p_{\rm NB}(t_{\rm ter},\vec{S}))\phi(t_{\rm ter},\vec{S})V(t_{\rm ter},\vec{S})\right| \nonumber \\
		&\le& e^{-rt_{\rm ter}}\int_{\mathbb{R}_+^d\setminus \hat{D}} d\vec{S} \phi(t_{\rm ter},\vec{S})p_{\rm NB}(t_{\rm ter},\vec{S})V(t_{\rm ter},\vec{S})+e^{-rt_{\rm ter}}\int_{\hat{D}\setminus D_{\rm half}} d\vec{S} (1-p_{\rm NB}(t_{\rm ter},\vec{S}))\phi(t_{\rm ter},\vec{S})V(t_{\rm ter},\vec{S}) \nonumber \\
		&\le& e^{-rt_{\rm ter}}\int_{\mathbb{R}_+^d\setminus D_{\rm half}} d\vec{S} \phi(t_{\rm ter},\vec{S})V(t_{\rm ter},\vec{S}) \nonumber \\
		&\le& \epsilon,
	\end{eqnarray}
	where we used Lemma \ref{lemma:VIntegBound} and $t_{\rm ter}<t_c$ at the last inequality.
	Combining these, we obtain (\ref{eq:VIntDiff}).
\end{proof}

\subsection{Proof of Lemma \ref{lem:finSumApp} \label{sec:PrLemFinSumApp}}

\begin{proof}
	The following holds
	\begin{eqnarray}
		\left|e^{-rT}\vec{p} \cdot \vec{\Tilde{Y}}(\tau_{\rm ter})-V_0\right|
		&\le& 
		e^{-rT}\left|\vec{p} \cdot \left(\vec{\Tilde{Y}}(\tau_{\rm ter})-\vec{Y}(\tau_{\rm ter})\right)\right|\nonumber \\
		&+&e^{-rT}\left|\vec{p} \cdot \vec{Y}(\tau_{\rm ter})-\int_{\Tilde{\Tilde{D}}} d\vec{x} \Tilde{\phi}(t_{\rm ter},\vec{x})Y(\tau_{\rm ter},\vec{x})\right|\nonumber \\
		&+&\left|e^{-rT}\int_{\Tilde{\Tilde{D}}} d\vec{x} \Tilde{\phi}(t_{\rm ter},\vec{x})Y(\tau_{\rm ter},\vec{x})-V_0\right|.
	\end{eqnarray}
	The first term can be evaluated as
	\begin{eqnarray}
		&& e^{-rT}\left|\vec{p} \cdot \left(\vec{\Tilde{Y}}(\tau_{\rm ter})-\vec{Y}(\tau_{\rm ter})\right)\right| \nonumber \\
		&\le& e^{-rT}\|\vec{p}\|\times \left\|\vec{\Tilde{Y}}(\tau_{\rm ter})-\vec{Y}(\tau_{\rm ter})\right\|\nonumber \\
		&<&e^{-rT}\frac{\sqrt{\prod_{i=1}^d\Delta_i}}{(4\pi)^{d/4}\sqrt{N_{\rm gr}}(\det\rho)^{1/4}}\sqrt{N_{\rm gr}}\frac{(4\pi)^{d/4}(\det\rho)^{1/4}}{\sqrt{\prod_{i=1}^d\Delta_i}}\epsilon \nonumber \\
		&=&e^{-rT}\epsilon\nonumber \\
		&<& \epsilon, \label{eq:tocyuu5}
	\end{eqnarray}
	where we used (\ref{eq:P2}), (\ref{eq:hxcond3}), and Lemma \ref{lem:FDMErr}.
	In order to bound the second term, we note that $\vec{p} \cdot \vec{Y}(\tau_{\rm ter})$ is an approximation of $\int_{\Tilde{\Tilde{D}}} d\vec{x} \Tilde{\phi}(t_{\rm ter},\vec{x})Y(\tau_{\rm ter},\vec{x})$ by the midpoint rule.
	Then, according to \cite{Linden},
	\begin{equation}
		\left|\vec{p} \cdot \vec{Y}(\tau_{\rm ter})-\int_{\Tilde{\Tilde{D}}} d\vec{x} \Tilde{\phi}(t_{\rm ter},\vec{x})Y(\tau_{\rm ter},\vec{x})\right| < \frac{1}{24}\left(\sum_{i=1}^d h_i^2\right) \times \left(\prod_{i=1}^d(u_i-l_i)\right) \times\eta
	\end{equation}
	holds under Assumption \ref{ass:phiYSm}, and therefore
	\begin{equation}
		e^{-rT}\left|\vec{p} \cdot \vec{Y}(\tau_{\rm ter})-\int_{\Tilde{\Tilde{D}}} d\vec{x} \Tilde{\phi}(t_{\rm ter},\vec{x})Y(\tau_{\rm ter},\vec{x})\right| < e^{-rT}\epsilon<\epsilon \label{eq:tocyuu6}
	\end{equation}
	under (\ref{eq:hxcond3}).
	The third term can be bounded as (\ref{eq:VIntDiff}) by Lemma \ref{lem:tter}.
	Combining (\ref{eq:tocyuu5}), (\ref{eq:tocyuu6}) and (\ref{eq:VIntDiff}), we obtain the claim.
\end{proof}

\subsection{Proof of Lemma \ref{lem:CompOnePsi} \label{sec:PrCompOnePsi}}

\begin{proof}
	Applying the algorithm in \cite{Berry2} to the ODE system (\ref{eq:ODE}) with $h_i$ satisfying (\ref{eq:hxcond3}), we obtain (\ref{eq:outputBerryAlgoFDM}).
	Since smaller $h_i$'s lead to larger $\|F\|$, and then larger complexity, we take as large $h_i$'s as possible, that is,
	\begin{equation}
		h_i = \frac{u_i-l_i}{\left\lceil\frac{u_i-l_i}{\Tilde{h}_i}\right\rceil} = \Theta(\Tilde{h}_i)
	\end{equation}
	Then, we can evaluate the complexity by substituting $s$ and $\|A\|$ in (\ref{eq:compQODE}) with the sparsity and norm of $F$, respectively.
	The sparsity of $F$ is $O(d^2)$, since the matrices constituting $F$ as (\ref{eq:F}) have sparsity at most 4 and the total number of them is $O(d^2)$. 
	Besides,
	\begin{equation}
		\|F\| = O\left(\max \left\{\frac{\sqrt{\prod_{i=1}^d \Delta_i}d^2 \Xi \sigma^2_{\rm max}}{(4\pi)^{d/4}(\det \rho)^{1/4}},d\eta\prod_{i=1}^d(u_i-l_i)\right\}\times \frac{d^2\sigma^2_{\rm max}}{\epsilon}\right) \label{eq:tocyuu7}
	\end{equation}
	for $h_i = \Theta(\Tilde{h}_i)$, as we will show soon.
	Using these, we obtain (\ref{eq:compQODEFDM}) by simple algebra.
	
	The remaining task is to show (\ref{eq:tocyuu7}).
	Note that
	\begin{eqnarray}
		\|F\|&\le&\sum_{i=1}^d\frac{\sigma_i^2}{2h_i^2} \|D^{\rm 2nd}\| + \sum_{i=1}^{d-1} \sum_{j=i+1}^d  \frac{\sigma_i\sigma_j}{4h_ih_j} \|D^{\rm 1st}\|^2 + \sum_{i=1}^d \frac{1}{2h_i}\left|r-\frac{1}{2}\sigma_i^2\right| \|D^{\rm 1st}\|
		\nonumber \\
		&=& \sum_{i=1}^d \frac{2\sigma_i^2}{h_i^2} + \sum_{i=1}^{d-1} \sum_{j=i+1}^d  \frac{\sigma_i\sigma_j}{h_ih_j} + \sum_{i=1}^d \frac{1}{h_i}\left|r-\frac{1}{2}\sigma_i^2\right|. \label{eq:tocyuu8}
	\end{eqnarray}
	Here, we used $\|D^{\rm 1st}\|=2, \|D^{\rm 2nd}\|=4$, which follows the fact that the eigenvalues of the $n\times n$ tridiagonal Toeplitz matrix
	\begin{equation}
	\begin{pmatrix}
		b & c &   &   &\\
		a & b & c &   &\\
		  &    \ddots &\ddots &\ddots & \\
		  &     &a &b & c\\
		  &           &  & a & b
	\end{pmatrix},
	\nonumber
	\end{equation}
	where $a,b,c\in\mathbb{C}$,	are $b+2\sqrt{ac}\cos \left(\frac{j\pi}{n+1}\right),j=1,...,n$\cite{Horn}.
	Then, substituting $h_i$ in (\ref{eq:tocyuu8}) by $\Tilde{h}_i$, we obtain (\ref{eq:tocyuu7}) by simple algebra.
\end{proof}

\subsection{Proof of Theorem \ref{th:main} \label{sec:PrThMain}}

\begin{proof}
	
	First, we show that, for $\epsilon_1$, $\epsilon_2$ and $\epsilon_{\Psi}$ satisfying (\ref{eq:eps1cond}), (\ref{eq:eps2cond}) and (\ref{eq:epsPsiCond}), respectively, Algorithm \ref{alg2} outputs $\omega$ such that (\ref{eq:finalErr}).
	For this, we begin with writing $\ket{\Tilde{\Psi}_{\rm mod}}$, the state which we obtain by applying the QLS algorithm of \cite{Childs3} to (\ref{eq:LSMod}), in the form of
	\begin{equation}
		\ket{\Tilde{\Psi}_{\rm mod}} := \frac{1}{\Tilde{Z}}\left(\ket{\Tilde{\Psi}_{\rm gar}}+\ket{\Tilde{\Psi}_1}+\ket{\Tilde{\Psi}_2}\right),
	\end{equation}
	where $\ket{\Tilde{\Psi}_{\rm gar}}$, $\ket{\Tilde{\Psi}_1}$ and $\ket{\Tilde{\Psi}_2}$ are the unnormalized states in the forms of
	\begin{equation}
		\ket{\Tilde{\Psi}_{\rm gar}} := \sum_{j=0}^{p(k+1)-1}\ket{j}\ket{\psi_j}, 
		\ket{\Tilde{\Psi}_1} := \sum_{j=p(k+1)}^{p(k+2)}\ket{j}\ket{\psi_j}, 
		\ket{\Tilde{\Psi}_2} := \sum_{j=p(k+2)}^{p(k+3)+1}\ket{j}\ket{\psi_j}, 
	\end{equation}
	with some unnormalised states, respectively, and $\Tilde{Z}:=\sqrt{\braket{\Tilde{\Psi}_{\rm gar}|\Tilde{\Psi}_{\rm gar}}+\braket{\Psi_1|\Psi_1}+\braket{\Psi_2|\Psi_2}}$.
	Because of (\ref{eq:tilPsiCond}), we see that
	\begin{equation}
		|\braket{\Pi|\Tilde{\Psi}_{\rm mod}}-\braket{\Pi|\Psi_{\rm mod}}|<\epsilon_{\Psi}.
	\end{equation}
	Then, since $E_1$, the output of the step 1 in Algorithm \ref{alg2}, satisfies
	\begin{equation}
		|E_1-\braket{\Pi | \Tilde{\Psi}_{\rm mod}}| < \epsilon_1,
	\end{equation}
	we obtain
	\begin{equation}
		|E_1-\braket{\Pi|\Psi_{\rm mod}}|<\epsilon_1+\epsilon_{\Psi}. \label{eq:E1Err}
	\end{equation}
	Similarly, since
	\begin{equation} \left\|\frac{1}{\Tilde{Z}}\ket{\Tilde{\Psi}_2}-\frac{1}{Z}\sum_{j=p(k+2)+1}^{p(k+3)+1}\ket{j}\ket{\vec{\gamma}}\right\|<\epsilon_{\Psi}
	\end{equation}
	because of (\ref{eq:tilPsiCond}) and
	\begin{equation}
		\left|E_2-\frac{\|\ket{\Tilde{\Psi}_2}\|}{\Tilde{Z}}\right| < \epsilon_2,
	\end{equation}
	we obtain
	\begin{equation}
		\left|E_2-\left\|\frac{1}{Z}\sum_{j=p(k+2)+1}^{p(k+3)+1}\ket{j}\ket{\vec{\gamma}}\right\| \ \right|=\left|E_2-\frac{\gamma\sqrt{(p+1)N_{\rm gr}}}{Z}\right|<\epsilon_2+\epsilon_{\Psi}. \label{eq:E2Err}
	\end{equation}
	Using (\ref{eq:E1Err}) and (\ref{eq:E2Err}), we see that $\omega:=e^{-rT}\gamma\sqrt{N_{\rm gr}}PE_1/E_2$ satisfies
	\begin{equation}
		\left|\omega-e^{-rT}\vec{p} \cdot \vec{\Tilde{Y}}(\tau_{\rm ter})\right|<\frac{e^{-rT}PZ}{\sqrt{p+1}}(\epsilon_{\Psi}+\epsilon_1)+\frac{e^{-rT}(\vec{p} \cdot \vec{\Tilde{Y}}(\tau_{\rm ter}))Z}{\gamma\sqrt{(p+1)N_{\rm gr}}}(\epsilon_{\Psi}+\epsilon_2), \label{eq:tocyuu9}
	\end{equation}
	by simple algebra.
	Here, note that, because of (\ref{eq:PsiTrashNorm}) and Assumption \ref{ass:gamma},
	\begin{equation}
		Z=O\left(g\sqrt{(p+1)N_{\rm gr}}\bar{Y}(\tau_{\rm ter})\right) = O\left(g\sqrt{(p+1)N_{\rm gr}}\gamma\right) \label{eq:tocyuu10}
	\end{equation}
	holds.
	Thus, if $\epsilon_1$, $\epsilon_2$ and $\epsilon_{\Psi}$ satisfy (\ref{eq:eps1cond}), (\ref{eq:eps2cond}) and (\ref{eq:epsPsiCond}), respectively, combining (\ref{eq:tocyuu9}), (\ref{eq:tocyuu10}) and (\ref{eq:P2}) leads to
	\begin{equation}
		\left|\frac{e^{-rT}\gamma\sqrt{N_{\rm gr}}PE_1}{E_2}-e^{-rT}\vec{p} \cdot \vec{\Tilde{Y}}(\tau_{\rm ter})\right|=O(\epsilon).
	\end{equation}
 	Finally, this and (\ref{eq:finSumUB}) yield (\ref{eq:finalErr}).
 	
 	Next, let us show that Algorithm \ref{alg2} with such $\epsilon_1$, $\epsilon_2$ and $\epsilon_{\Psi}$ has the complexity (\ref{eq:compFinal}).
 	We just multiply the complexity of generating $\ket{\Tilde{\Psi}_{\rm mod}}$ once, which is given by (\ref{eq:compQODEFDM}) with $\epsilon^\prime=\epsilon_\Psi$, by the number of the generation, which is $O(\max\{1/\epsilon_1,1/\epsilon_2\})$ since the QAE with $O(1/\delta)$ queries outputs the estimation with the error of $O(\delta)$.
 	By simple algebra, we obtain (\ref{eq:compFinal}).
	
\end{proof}

\end{document}